\documentclass[letter,12pt]{article}
\pdfoutput=1

\usepackage{jheppub} % for details on the use of the package, please
\usepackage[T1]{fontenc} % if needed
\usepackage{amsthm}
\usepackage{slashed}
\usepackage{epsfig}
\usepackage{comment}
\usepackage{cancel}
\usepackage{bbm}
\usepackage{array}
\usepackage{bigints}
\usepackage{booktabs}
\usepackage{color}
\usepackage{dsfont}
\usepackage{float}
\usepackage{framed}
\usepackage{graphicx}
\usepackage{indentfirst}
\usepackage[mathscr]{eucal}
\usepackage{multirow}
\usepackage{tikz}
\usetikzlibrary{arrows.meta}
\usepackage{subdepth}
\usepackage{titlesec}
\usepackage[dotinlabels]{titletoc}
\usepackage{wrapfig}
\usepackage[all]{xy}
\usepackage[vcentermath]{youngtab}
\usepackage{relsize}
\usepackage{hyperref}
\usepackage{makecell}  % VH added for table line breaks
\usepackage[table]{colortbl}

\numberwithin{equation}{section}

\usepackage[utf8]{inputenc}
\usepackage{slashed}
\usepackage{amsmath}
\usepackage{amsfonts}
\usepackage{amssymb}
\usepackage{mathtools}

\usepackage{cleveref}
\crefname{section}{§}{§§}
\Crefname{section}{§}{§§}

\newtheorem{theorem}{Theorem}

\def\0{{(0)}}
\def\1{{(1)}}
\def\2{{(2)}}

\def\<{\langle }
\def\>{\rangle }

\newcommand{\bea}{\begin{eqnarray}}
\newcommand{\eea}{\end{eqnarray}}
\newcommand{\be}{\begin{equation}}
\newcommand{\ee}{\end{equation}}
\newcommand{\ba}{\begin{align}}
\newcommand{\ea}{\end{align}}

   \makeatletter
  \let\over=\@@over \let\overwithdelims=\@@overwithdelims
  \let\atop=\@@atop \let\atopwithdelims=\@@atopwithdelims
  \let\above=\@@above \let\abovewithdelims=\@@abovewithdelims
\renewcommand\section{\@startsection {section}{1}{\z@}%
                                   {-3.5ex \@plus -1ex \@minus -.2ex}%nn
                                   {2.3ex \@plus.2ex}%
                                   {\normalfont\large\bfseries}}

\renewcommand\subsection{\@startsection{subsection}{2}{\z@}%
                                     {-3.25ex\@plus -1ex \@minus -.2ex}%
                                     {1.5ex \@plus .2ex}%
                                     {\normalfont\bfseries}}

\newcommand{\beq}{\begin{equation}}
\newcommand{\eeq}{\end{equation}}
\newcommand{\beqa}{\begin{eqnarray}}
\newcommand{\eeqa}{\end{eqnarray}}
\newcommand{\beqar}{\begin{eqnarray*}}

\def\[{\big[}
\def\]{\big]}

\def\min{\text{min}}
\def\a{\alpha}
\def\b{\beta}
\def\G{\Gamma}

\def\a{{\alpha}}
\def\b{{\beta}}

\def\be{{\bar \epsilon}}

\def\CA{{\mathcal A}}
\def\CB{{\mathcal B}}
\def\CC{{\mathcal C}}

\def\CX{{\mathcal X}}

\def\SB{{\mathscr B}}

\def\+{{(+)}}
\def\-{{(-)}}
\def\0{{(0)}}
\def\1{{(1)}}
\def\2{{(2)}}
\def\3{{(3)}}
\def\4{{(4)}}
\def\5{{(5)}}

\def\n{{({\sf N})}}
\def\i{{\{i\}}}

\def\S{{\tt{S}}}
\def\PT{\text{PT}}
\def\PTt#1{{$\text{PT}_{{#1}}$}}  % a 2s-PT when written in text (not in math mode)
\def\SPT#1{{\vec \S_\PT^{\,\{#1\}}}}
\def\k{{\tt{K}}}
\def\l{\lambda}
\def\ovx#1{\vec x^{\, \{#1\}}}
\def\ovy#1{\vec y^{\, \{#1\}}}

\def\N{{\sf N}}
\def\D{{\sf D}}
\def\R{{\sf R}}

\def\powerset{\mathscr{P}}

\def\polyindfont#1{\mathscr{#1}}
\def\pI{\polyindfont{I}}
\def\pJ{\polyindfont{J}}
\def\pK{\polyindfont{K}}

\def\ii{{\underline{i}}}
\def\jj{{\underline{j}}}
\def\kk{{\underline{k}}}
\def\pII{{\underline{\pI}}}

\def\pKK{{\underline{\pK}}}

\def\gI{\Gamma}
\def\gInz{{\tilde \Gamma}}

\def\SB#1{{\hat e}^{\, #1}}
\def\IB#1{{\hat f}^{\, #1}}
\def\KB#1{{\hat g}^{\, #1}}

\def\MStoI{{\bf M}}
\def\MStoK{\widetilde {\bf M}}

\def\IQfont#1{{\mathbf #1}}
\def\Sv#1{{\vec{\S}}^{\,(#1)}}  % entropy vector
\def\Q{{\IQfont{Q}}}

\def\I{{\tt I}}

\newcommand{\extr}[1]{\mathcal{E}_{#1}}

\definecolor{vecolor}{rgb}{0.7,0.3,0.9}

\def\RR{\mathbb{R}}
\def\shadeI{\cellcolor{blue!5}}
\def\shadeK{\cellcolor{red!5}}
\def\sref#1{\textsection\ref{#1}}

\title{\boldmath Holographic Entropy Relations Repackaged}
 
\author[a]{Temple He}
\author[b]{, Matthew Headrick}
\author[a]{, Veronika E.\ Hubeny}

\affiliation[a]{Center for Quantum Mathematics and Physics (QMAP)\\
Department of Physics, University of California, Davis, CA 95616 USA}
\affiliation[b]{Martin Fisher School of Physics, Brandeis University, Waltham MA 02453, USA}

\emailAdd{tmhe@ucdavis.edu}
\emailAdd{headrick@brandeis.edu}
\emailAdd{veronika@physics.ucdavis.edu}

\preprint{BRX-TH-6652}

\abstract{
We explore the structure of holographic entropy relations (associated with  `information quantities' given by a linear combination of entanglement entropies of spatial sub-partitions of a CFT state with geometric bulk dual).  
Such entropy relations can be recast in multiple ways, some of which have significant advantages. 
Motivated by the already-noted simplification of entropy relations when recast in terms of multipartite information, we explore additional simplifications when recast in a new basis, which we dub the K-basis, constructed from perfect tensor structures.  For the fundamental information quantities such a recasting is surprisingly compact, in part due to the interesting fact that entropy vectors associated to perfect tensors are in fact extreme rays in the holographic entropy cone (as well as the full quantum entropy cone).
More importantly, we prove that all holographic entropy inequalities have positive coefficients when expressed in the K-basis, underlying the key advantage over the entropy basis or the multipartite information basis.  
}

\begin{document} 
\maketitle
\flushbottom

\section{Introduction}

In recent years, entanglement has played an increasingly prominent role in holography.  Early hints that entanglement structure of a CFT state may elucidate the emergence of bulk spacetime in its dual description stemmed from the conjecture of Ryu and Takayanagi (RT) \cite{Ryu:2006bv} (and its covariant generalization by Hubeny-Rangamani-Takayanagi (HRT) \cite{Hubeny:2007xt}), which recasts entanglement entropy of a spatial region in terms of an extremal surface area in the bulk.\footnote{\
In particular, for any `geometric' CFT state whose bulk dual is characterized by a classical geometry (with arbitrary time dependence), the entanglement entropy of any spatial region $\CA$ is given by quarter-area of the smallest area extremal surface $\extr{\CA}$ homologous to $\CA$, namely $\S(\CA) = \frac{1}{4}\, \text{Area}(\extr{\CA})$.
}
Such `geometrization' of entanglement has been tremendously useful in gaining further insight into various crucial properties of entanglement entropy.
Moreover, this relation is intriguingly reminiscent of the relation between black hole entropy and its event horizon area, perhaps harking to a deeper principle yet to be fully appreciated, and recently motivated bolder conjectures concerning the link between entanglement and the geometry of spacetime in quantum gravity \cite{VanRaamsdonk:2010pw,Maldacena:2013xja}; see e.g.\ \cite{Rangamani:2016dms} for a review.

To probe this connection further, it is fruitful to subdivide the full system into multiple subsystems and consider the relations between the entanglement entropies of all possible combinations of these subsystems.  A particularly natural set of relations takes the form of an inequality between sums of entanglement entropies.  For example, given any two subsystems $\CA$ and $\CB$ of the boundary CFT, and denoting their union $\CA \cup \CB \equiv \CA\CB$, their entropies must satisfy the relation known as subadditivity (SA):
\begin{align}
	\S(\CA) + \S(\CB) \geq \S(\CA\CB)\ .
\label{e:SA}
\end{align}
This relation holds for any quantum system (in any state and any subdivision into $\CA$ and $\CB$ with Hilbert space factorization $\mathcal{H}= \mathcal{H}_{\CA} \otimes   \mathcal{H}_{\CB} \otimes  \cdots $),\footnote{\
Strictly speaking, the Hilbert space on the boundary CFT does not necessarily factorize, and a careful study requires Tomita-Takasaki theory. However, as nicely summarized in \cite{Witten:2018lha}, oftentimes one may still get the correct result by (incorrectly) assuming this simple factorization.
} and is equivalent to the positivity of mutual information, i.e.
\begin{align}
\I_2(\CA\!:\!\CB) \equiv \S(\CA) + \S(\CB) - \S(\CA\CB)
\label{e:MI}\ ,
\end{align}
which characterizes the total amount of correlation between $\CA$ and $\CB$.  
A stronger relation, pertaining to a subdivision into three subsystems labeled $\CA$, $\CB$, and $\CC$, known as strong subadditivity (SSA) is given by
\begin{align}
	\S(\CA\CB) + \S(\CB\CC) \geq  \S(\CB) + \S(\CA\CB\CC)\ ,
\label{e:SSA}
\end{align}
which is equivalent to the statement of monotonicity of mutual information under inclusion,  $\I_2(\CA\!:\!\CB\CC) \ge \I_2(\CA\!:\!\CB) $, and likewise holds universally.

However, it is even more interesting to consider entropy relations which are not satisfied universally for all quantum states, but are satisfied in any 
 `geometric state' of a holographic CFT\footnote{\
By a holographic CFT we mean any CFT that admits a higher-dimensional gravitational dual.  Within such a CFT, we define a `geometric state' as one whose dual is described in terms of classical bulk geometry. In particular, we work in the large central charge and large 't Hooft coupling limit, where the quantum and stringy effects are suppressed.
}
when the partitions in question are spatial regions (on a given Cauchy slice of the background spacetime on which the CFT lives).
The simplest example of such a relation is the monogamy of mutual information (MMI)
\begin{align}
\begin{split}
	\S(\CA\CB) + \S(\CA\CC) + \S(\CB\CC) \geq \S(\CA) + \S(\CB) + \S(\CC) + \S(\CA\CB\CC)\ ,
\label{e:MMI}
\end{split}
\end{align} 
which can be re-expressed as negativity of the tripartite information,
\begin{align}
\I_3(\CA\!:\!\CB\!:\!\CC) \equiv \S(\CA) + \S(\CB) + \S(\CC) - \S(\CA\CB) - \S(\CA\CC) - \S(\CB\CC) + \S(\CA\CB\CC) \ .
\label{e:TI}
\end{align}
Both SSA and MMI are relatively easy to prove holographically\footnote{\
For the static case this follows immediately from the definition of a minimal surface  \cite{Headrick:2007km,Hayden:2011ag} while in the covariant case it requires extra assumptions (such as the null energy condition) \cite{Wall:2012uf}.
In the case of SSA, these restricted proofs are substantially easier than the full proof \cite{Lieb:1973cp} pertaining to a general quantum system.
}
 for geometric states (in fact, using almost identical arguments, which belies their fundamental difference\footnote{\
Nevertheless, alternate proofs \cite{Hubeny:2018bri,Cui:2018dyq} utilizing a rather different geometrization of entanglement entropy in terms of \emph{bit threads} \cite{Freedman:2016zud} do reveal an important difference between the nature of the two inequalities despite their superficial similarities.  Although in the following we will not evoke these methods, it is noteworthy that the approach of  \cite{Cui:2018dyq}  offered hints at the utility of perfect tensor basis that we will be focusing on in the present work.
}). 

Further refinement into a larger number of subsystems  $\CA_1,\ldots,\CA_\N$ then yields further interesting holographic entropy relations.  
Such a study was initiated in \cite{Bao:2015bfa}, where the authors developed an algorithm for discovering new holographic entropy inequalities.  In the entropy space (defined below) the set of all inequalities delineates a \emph{holographic entropy cone}, which characterizes the restriction on entanglement structure of any physically allowed geometric state.  More recently this study was reinvigorated by the program originated in \cite{Hubeny:2018trv} and further developed in \cite{Hubeny:2018ijt}.   While  \cite{Bao:2015bfa} characterized the entropy cone using its extreme rays,\footnote{\
The extreme rays are the minimal set of rays whose convex hull is the cone.} the latter approach instead focuses on the hyperplanes in the entropy space, corresponding to information quantities such as \eqref{e:MI} and \eqref{e:TI}.  The full set of so-called primitive\footnote{\label{f:primitive}
A {\it primitive} information quantity $\Q$ is defined as one for which there exists a geometric state and a configuration of regions satisfying $\Q=0$ while simultaneously having a nonzero value for any other independent information quantity.  Physically, this ensures that the information quantity captures a form of correlation that can vanish and is independent of all the others in the arrangement.
} information quantities is dubbed the \emph{holographic entropy arrangement}, while the intersection of the half-spaces given by the sign-definite ones defines the \emph{holographic entropy polyhedron} which provides an explicit construction of the holographic entropy cone.\footnote{\
Strictly-speaking, the equivalence between the holographic entropy cone and the holographic entropy polyhedron for all geometric states presently remains a conjecture, albeit a strong one \cite{Hubeny:2018ijt}.  
}

However, the entropy space itself is exponentially large -- for $\N$ subsystems, we can form $\D=2^\N - 1$ independent entropies corresponding to the various composite subsystems, so the entropy space lives in $\RR^\D$.  This means that we should expect the polyhedron to have a huge number of facets. It is therefore extremely useful to take advantage of the inherent symmetries of the setup, which will simplify our considerations substantially.   In particular, the full structure of the holographic entropy arrangement, as well as the polyhedron itself, must be symmetric under permutation and purification symmetry, extensively discussed in \cite{Hubeny:2018ijt}. 
In the next few paragraphs we will briefly review both of these in turn, to pave the way for our main consideration, namely one of representing information quantities.

The obvious symmetry  of our constructs (arrangement and polyhedron) is the permutation of region names (since physics cannot depend on our naming conventions). 
As an illustrative example, consider the two information quantities \eqref{e:MI} and \eqref{e:TI}, namely  $\I_2(\CA\!:\!\CB)$  and  $\I_3(\CA\!:\!\CB\!:\!\CC)$. Each is manifestly symmetric in permuting their arguments.  Similarly, the natural generalization to $\N$ parties, known as multipartite information $\I_\N(\CA_1\!:\!\ldots\!:\!\CA_\N)$ (defined below in \eqref{e:ItofromS}) has a manifest $\mathbf{S}_\N$ symmetry.  However, in the case of e.g.\ $\N=3$ parties, $\I_2(\CA\!:\!\CB)$ does not have the full $\mathbf{S}_3$ symmetry, which means that in addition to SA in the form \eqref{e:SA} we will automatically also have two further versions of SA, obtained by replacing $\CA$  or  $\CB$  with $\CC$.  These three inequalities lie in the same symmetry orbit.\footnote{\label{f:nonprimitive}
A priori, in the 3-party case we could also consider another version of \eqref{e:SA} by combining regions, such as $\I_2(\CA\!:\!\CB\CC)\ge 0$ and its two inequivalent permutations; however, these quantities are not primitive, being sums of non-negative quantities (e.g. $\I_2(\CA\!:~\!\CB\CC)= \I_2(\CA\!:~\!\CB)+\I_2(\CA\!:~\!\CC)-\I_3(\CA\!:~\!\CB\!:~\!\CC)$) and so cannot vanish without the individual components vanishing.}
To describe the full arrangement structure, it then suffices to specify just a single representative for each orbit, and permute the labels to generate the remaining ones.  

A more subtle symmetry of both the arrangement and the polyhedron is the purification symmetry.  In particular, for $\N$ subsystems $\CA_1,\ldots,\CA_\N$, we define its \emph{purifier} as the complement subsystem $\CA_{\N+1} \equiv \left( \CA_1\cdots\CA_\N \right)^c$, so that the total state on $\CA_1\cdots\CA_{\N+1}$ is pure.\footnote{\
In our holographic setup, this assumes a pure state on the entire boundary spacetime on which the CFT lives.  For instance, in case of the thermofield double state describing an eternal Schwarzschild-AdS black hole, with $\CA_1,\ldots,\CA_\N$ located on just the right boundary, the purifier region $\CA_{\N+1}$ would contain the remainder of the right boundary space as well as the entire left boundary space.
}  The entanglement entropy then satisfies $\S(\CA_1\cdots\CA_{\N+1})=0$, and for any subpartition, characterized by $n=1, \ldots, \N$,
\begin{align}
\S(\CA_1\cdots\CA_{n})=\S(\CA_{n+1}\cdots\CA_{\N+1}) \ .
\label{e:Spurif}
\end{align}
We can now treat $\CA_{\N+1}$ on equal footing with the other $\N$ subsystems, and employ the larger permutation symmetry $\mathbf{S}_{\N+1}$.  However, even though this symmetry is to be viewed as fundamentally the same as the permutation symmetry, the information quantities it generates can take on more apparently distinct form since our entropy space does not manifest this full symmetry when expressed in the usual entropy basis.

As an illustrative example of generating new relations via purification, consider SA in the $\N=2$ case, and let $\CC\equiv \left(\CA\mathcal{B}\right)^c$ be the purifier.  Replacing $\S(\CB)=\S(\CA\CC)$ and $\S(\CA\CB)=\S(\CC)$ in \eqref{e:SA}, and then renaming $\CC$ back to $\CB$, generates the Araki-Lieb (AL) inequality
\begin{align}
	\S(\CA) + \S(\CA\CB)  \geq  \S(\CB)\ ,
\label{e:AL}
\end{align}
which is therefore likewise an $\N=2$ entropy inequality, and indeed holds universally for all quantum states.\footnote{\
Similarly, for $\N=3$, purifying SSA \eqref{e:SSA} with respect to $\CA$ and permuting the labels yields the weak monotonicity (WM), $\S(\CA\CB) + \S(\CB\CC) \geq  \S(\CA) + \S(\CC)$.  However, neither SSA nor WM are primitive for the same reason as explained in footnote \ref{f:nonprimitive}, and therefore do not form the facets of the holographic entropy polyhedron.}

Having explained the full $\mathbf{S}_{\N+1}$ symmetry of the holographic entropy arrangement (and correspondingly the polyhedron and the cone), let us briefly review what is known so far about these constructs.\footnote{\
We refer the reader to Appendix \ref{a:QinSIKbases} for the explicit forms of these entropy inequalities and information quantities.
}
For  $\N=2$, the holographic entropy cone, which lives in $\RR^3$, is specified by SA \eqref{e:SA} and two permutations of AL  \eqref{e:AL}, and so here it actually coincides with the full quantum entropy cone.
For the $\N=3$ cone (in $\RR^7$), in addition to the six primitive quantities uplifted from the $\N=2$ cone, we also have one MMI \eqref{e:MMI}, which itself is fully permutation and purification symmetric.  The $\N=4$ cone similarly consists of  the various uplifts of $\N=2$ and $\N=3$ relations, along with the requisite permutations and purifications, and contains no new entropy inequalities. However, the $\N=4$ holographic entropy {\it arrangement} \cite{Hubeny:2018ijt} contains several additional information quantities such as $\I_4$, though none of these are sign-definite (meaning there exist holographic configurations $\CA_1,\ldots,\CA_4$ for which $\I_4(\CA_1\!:\!\cdots\!:\!\CA_4)<0$ and other holographic configurations $\CA'_1,\ldots,\CA'_4$ for which $\I_4(\CA'_1\!:\!\cdots\!:\!\CA'_4)>0$).
The full holographic entropy cone is now also known for $\N=5$ \cite{Cuenca:2019uzx}, and consists of specific uplifts of the lower-$\N$ inequalities along with five new 5-party inequalities initially constructed by \cite{Bao:2015bfa}.\footnote{\
While \cite{Bao:2015bfa} was only able to provide an upper bound on the cone by proving the five new 5-party inequalities, but unable to realize all the extreme rays of this cone by explicit configurations (leading them to conjecture two further inequalities which would shrink the cone enough to realize all of its extreme rays), recently \cite{Cuenca:2019uzx} nailed the  $\N=5$ cone fully by explicitly realizing all extreme rays using the approach of \cite{Hubeny:2018trv,Hubeny:2018ijt}.  
Correspondingly, the two conjectured inequalities can be explicitly violated; we thank Xi Dong and Sergio Hernandez Cuenca for alerting us to their (independent) counter-examples.
} 
When written out as individual inequalities, the cone (which lives in $\RR^{31}$) now has 372 facets, but they are organized into just eight separate symmetry orbits \cite{Cuenca:2019uzx}, five of which correspond to the new inequalities.

Although SA and MMI have a natural quantum information theoretic interpretation, the five new $\N=5$ entropy inequalities, as well as the other (sign-indefinite) new information quantities, look rather more obscure and formidable, especially when written out explicitly in the entropy basis.  However, we have seen above that both SA and MMI simplify dramatically when re-expressed in terms of multipartite informations to just single-term expressions, $\I_2(\CA_i\!:\!\CA_j)\ge0$ and  $-\I_3(\CA_i\!:\!\CA_j\!:\!\CA_k)\ge0$, respectively.
In fact, a similar simplification occurs for the sign-indefinite information quantities.  Indeed, as proved in \cite{Hubeny:2018trv} for any $\N$, one nontrivial $\N$-party information quantity is precisely the multipartite information $\I_\N(\CA_1\!:\!\cdots\!:\!\CA_\N)$.
Moreover, for $\N=4$, two of the new information quantities \cite{Hubeny:2018ijt} have eight terms when expressed in the \emph{S-basis} (i.e.\ in terms of entanglement entropies), but only two terms when expressed in the \emph{I-basis} (i.e.\ in terms of multipartite information).  Most of the other information quantities likewise look substantially simpler in the I-basis than in the S-basis, though the precise reduction in the number of terms varies (cf.\ Appendix~\ref{a:QinSIKbases}).

This simplification indicates that there is something `nice' (or more fundamental) about the multipartite informations as opposed to the entanglement entropies of composite subsystems.  One reason highlighted in  \cite{Hubeny:2018ijt} was that of `balance', rooted in the UV structure of the expression.  In a local QFT, the entanglement entropy of any finite region diverges, whereas multipartite information of disjoint regions remains finite (since the individual UV divergences cancel out).\footnote{\
We define two regions to be disjoint when their closures don't intersect.
However, as the regions get closer to each other, their mutual information grows, and diverges when the regions touch (with their boundaries tangent somewhere).
}  In fact,  \cite{Hubeny:2018ijt} introduced a more refined version dubbed $\R$-balance,\footnote{\label{f:Rbalance}
An information quantity is $\R$-balanced if its expansion in the I-basis only contains $\I_n$'s with $n>\R$.
}
which provides a useful $\mathbf{S}_\N$-invariant organizational principle.  Nevertheless, given that the purifier of disjoint regions is necessarily adjoining to them, and therefore spoils the balance, this feature by itself cannot be the whole story.

However, above we have  seen another advantage of the $\I_n$'s, namely that they are manifestly permutation symmetric in their $n$ arguments.  Although a given  information quantity composed of multiple $\I_n$'s may lose some or all of the permutation symmetry, its substructure is better indicated when grouped into these permutation-symmetric components, as the geometrical relation of the facets depends on this substructure. Assuming this symmetry feature indeed underlies the simplification of the information quantities when expressed in the I-basis, it is then natural to wonder whether we can do better.  In particular, given that the full symmetry for the $\N$-party holographic entropy arrangement is the $\mathbf{S}_{\N+1}$  involving purifications as well as permutations, rather than just the permutation group $\mathbf{S}_\N$,  one might expect that an even better packaging should be one that explicitly takes advantage of this larger symmetry.  Specifically, we want to decompose the information quantities into substructures that treat all $\N+1$ parties (including the purifier) on equal footing.  One such particularly natural packaging, already utilized in \cite{Cui:2018dyq}, is in terms of the so-called perfect tensors (PT).  A perfect tensor (defined more accurately in \sref{ss:Kbasis}) describes a state correlating an even number of parties whose reduced density matrix for an any bipartition is maximally mixed within the PT substructure. As explained below, the collection of such structures can be used to construct a basis for the entropy space.\footnote{\
We want to thank Ning Bao for bringing the possible existence of such a basis to our attention.}  
We will call such a basis the \emph{K-basis} (named simply based on the notational nomenclature of our constructs) and explore the structure of information quantities when expressed in this basis.

It will turn out that, as expected, many of the information quantities (especially the ones which are more naturally obtained via purification) are much more compactly packaged in this basis.  However, while there remains a pronounced difference between the S-basis and the K-basis, there is less of a pronounced difference between the I-basis and the  K-basis  -- in fact, there is a close relation between them.   Indeed, some of the information quantities are still more compact in the I-basis than in the K-basis. Nevertheless, the K-basis has a major advantage compared to the I-basis, in that every sign-definite (non-negative) information quantity turns out to be expressed as a sum of terms with only non-negative coefficients in the K-basis! We will explicitly prove that this is always the case for all holographic entropy inequalities involving arbitrary number of parties.\footnote{\
An important assumption used in the proof is expounded in footnote~\ref{f:metric_cone}.
} 

The plan of the paper is as follows.  
In  \sref{s:setup} we will explain the setup and notation, and specify the three mentioned bases: S-basis, I-basis, and  K-basis.
We will then discuss in \sref{s:infoQ} the information quantities in all three bases, where we exemplify the recasting of various information quantities for $\N=3,4,5$ in these bases. We will also observe some general properties of the K-basis, but will relegate the proofs of such properties to \sref{positivity}.
Next, in  \sref{positivity}, we will prove that the holographic entropy inequalities in the new K-basis only involve positive linear combinations of the basis elements, as well as the fact the K-basis elements are all proportional to extreme rays. Finally, we will discuss future directions in \sref{discussion}.
We collect details of the various information quantities for $\N=4,5$ in Appendix~\ref{a:QinSIKbases}. In Appendix~\ref{a:contraction}, we prove a weaker form of Theorem~\ref{positivethm}, but this time by using the \emph{method of contraction} introduced in \cite{Bao:2015bfa} in order to demonstrate the method's utility.
%

%_________________________________________
\section{Entropy Space in S, I, and K Bases}
\label{s:setup}
%-------------------------------------------------------

Consider an $\N$-party system, which can be expressed as an $(\N+1)$-party pure state. We label the parties as $\CA_1,\ldots,\CA_{\N}$, and use $\CA_{\N+1}$ to denote the purification system.  When expressing our quantities using the standard S-basis in \sref{ss:Sbasis} or the I-basis in \sref{ss:Ibasis}, however, we will not evoke the purifier $\CA_{\N+1}$ explicitly; it will make its appearance only in the K-basis in \sref{ss:Kbasis}.

The traditional way to characterize the entanglement structure of such a system is to specify the entanglement entropy of each subsystem.  
Out of the collection $\{ \CA_i \}$ with $i \in [\N]\equiv \{1,2,\ldots,\N\}$, we can form 
$\D = 2^\N - 1$ independent systems by grouping these together in all possible ways.
Following the terminology of \cite{Hubeny:2018trv,Hubeny:2018ijt}, 
we will call the `atomic' indecomposable partitions like $\CA_i$ {\it monochromatic} subsystems, and composite ones  like $\CA_i\CA_j$ {\it polychromatic} subsystems (which also include the monochromatic ones as a special case).
For each polychromatic subsystem we can define the corresponding entanglement entropy; 
to simplify notation, we will use a shorthand to denote the entanglement entropy of a polychromatic subsystem by a subscript identifying the parties involved, e.g. 
$\S(\CA_{i})\equiv \S_{i}$, $\S(\CA_{i}\CA_{j}) \equiv  \S_{ij}$, etc.
When we wish our expressions to be even more compact, we will denote a generic polychromatic subsystem by a collective index $\pI \in \powerset([\N]) \setminus \emptyset$ (where $\powerset$ denotes the power set) ranging over all non-empty subsets of $[\N]$, and correspondingly denote its entropy simply as $\S_\pI$.\footnote{\label{f:purifier_index}
Later (particularly in the context of the K-basis) where we consider the $(\N+1)$-party pure state and treat the purifier on equal footing, we will use analogous indices for the monochromatic and polychromatic systems, but underlined to emphasize the difference in the context, namely $\ii=1,\ldots,\N+1$ and $\pII$ having the correspondingly extended range. }

The space of all $\D$ entanglement entropies is called the {\it entropy space}.
Given a specified quantum system, namely the full density matrix and a particular decomposition into $\N$ parties, its entanglement structure is characterized by the set $\{ \S_\pI \} $, which corresponds to a point in the entropy space; it is however more useful to think of it as a vector (from the origin to the specified point), the so-called {\it entropy vector}.  For an $\N$-party system, we will denote its entropy vector by $\Sv{\N}$.

A given vector, such as  $\Sv{\N}$, can be expressed in various bases, which describe the entanglement structure and other quantities of interest in terms of different constructs.  Three particularly convenient ones are:
\begin{itemize}
\item S-basis, which uses entanglement entropies of polychromatic subsystems
\item  I-basis, which uses multipartite informations between monochromatic subsystems
\item K-basis, which uses perfect tensor structures between monochromatic subsystems plus the purifier
\end{itemize} 
We now detail each of these in turn, before turning to describe information quantities expressed in these bases in \sref{s:infoQ}.

% _ _ _ _ _ _ _ _ _ _ _ _ _ _
\subsection{S-Basis}
\label{ss:Sbasis}
% - - - - - - - - - - - - - - - -

The most obvious basis for the entropy space is the entropy basis itself, which we call the S-basis.  This is the collection ordered by size of the polychromatic subsystem (and numerically within each):
\begin{align}\label{Sbasis}
	\Sv{\N} \equiv \left.\left\{\left\{\S_i\right\},\; \left\{\S_{ij}\right\},\;\ldots,\; \S_{12\cdots \N} \right\}\right|_{i,j,\ldots=1,\ldots,\N} \ ,
\end{align}
where $i < j$ and the subscripts are numerically ordered. 
As presaged above, the total number of components
of $\Sv{\N}$ is $\D=2^\N-1$.\footnote{\
For each party $i=1,\ldots,\N$ a given polychromatic system can either include it or not include it, which gives $2^\N$ possible combinations; however, the one with no party being included is disallowed and so we are left with $\D=2^\N-1$ possibilities.
Another way to see this is to use the grouping of \eqref{Sbasis} and 
 simple combinatorics:
$
	\sum_{r=1}^{\N}\binom{\N}{r} = 2^{\N}-1
$.
}
For instance, for $\N=3$, the entropy vector is
\begin{align}\label{Sbasis3}
	\Sv{3} = \left\{\S_1,\;\S_2,\;\S_3,\;\S_{12},\;\S_{13},\;\S_{23},\; \S_{123} \right\}\ .
\end{align}
We can expand this in terms of unit vectors 
$\SB{\pI}$   
pointing along the corresponding $\S_\pI$ axes: 
\begin{align}\label{Sbasis3ext}
	\Sv{3} = \sum_{\pI} \S_\pI \, \SB{\pI} 
	= \S_1 \, \SB{1} + \S_2 \,  \SB{2} + \S_{3}  \, \SB{3} + \S_{12}  \, \SB{12} + \S_{13}  \, \SB{13} + \S_{23} \,  \SB{23} + \S_{123}  \, \SB{123} \ ,
\end{align}
where $\SB{\pI}$ is a basis vector in the $7$-dimensional entropy space with $1$ in the $\pI$-th component and 0 elsewhere.

% _ _ _ _ _ _ _ _ _ _ _ _ _ _
\subsection{I-Basis}
\label{ss:Ibasis}
% - - - - - - - - - - - - - - - -

We can repackage the information in $\Sv{\N}$ by using the fact that there is a bijection between the polychromatic entropies $\S_\pI$ and multipartite informations with monochromatic arguments, $\I_n(\CA_1\!:\!\cdots\!:\!\CA_n)$.  In particular, denoting the cardinality (i.e.\ number of monochromatic components) of a polychromatic subsystem $\pI$ as $n_\pI$, we have the following conversion between $\I$ and $\S$:
\begin{equation}
\I_\pI=\sum_{\pK\subseteq\pI} (-1)^{1+n_\pK} \, \S_\pK
\qquad \Longleftrightarrow \qquad   
\S_\pI=\sum_{\pK\subseteq\pI} (-1)^{1+n_\pK} \,  \I_\pK \ ,
\label{e:ItofromS}
\end{equation}
where we have used the shorthand $\I_\pI$ to indicate the $n_\pI$-partite information with arguments given by the $n_\pI$ monochromatic subsystems composing $\pI$.
If we use a slightly less compact but more explicit shorthand\footnote{\
The arguments are relegated to the subscript and we drop the original subscript $n$ on $\I_n$ since this can be read off by simply counting the number of arguments.  There is no separation between the subscripts since for our basis we only consider multipartite information between monochromatic subsystems (and we can work in base $\N+1$ to have a single character for each $i$).  On the other hand, when we do wish to evoke multipartite information with polychromatic systems as its arguments, we will simply revert to the original non-shorthand notation.
}  
$\I_2(\CA_i:\CA_j) \equiv \I_{ij}$,  $\I_3(\CA_i: \CA_j: \CA_k) \equiv \I_{ijk}$, 
etc., then written out explicitly  for $n_\pI = 1,2,3,4$,  we have
\small
\begin{align}
\begin{split}
& \I_i 
	= \S_i  \\ 
& \I_{ij} 
	= \S_i+\S_j-\S_{ij} \\ 
& \I_{ijk} 
	= \S_i+\S_j+\S_k-\S_{ij}-\S_{ik}-\S_{jk}+\S_{ijk} \\ 
&  \I_{ijkl} 
	= \S_i+\S_j+\S_k+\S_l-\S_{ij}-\S_{ik}-\S_{il}-\S_{jk}-\S_{jl}-\S_{kl}+\S_{ijk}+\S_{ijl}+\S_{ikl}+\S_{jkl}-\S_{ijkl} \ ,
\end{split}
\label{e:ItoS4}
\end{align}
\normalsize
as well as the identical set of equations going the other way, with each $\S$ replaced by the corresponding $\I$ and vice versa.

We can then re-express the entropy vector in the I-basis as the ordered collection of $\I_\pI$'s.  For example, for $\N=3$ we have, analogously to \eqref{Sbasis3},
\begin{align}\label{Ibasis3}
	\Sv{3} = \left\{\I_1,\;\I_2,\;\I_3,\;\I_{12},\;\I_{13},\;\I_{23},\;  \I_{123}
	\right\}_{[\I]}  \ .
\end{align}
However, when writing it in this form, we have implicitly switched to the I-basis (indicated by the subscript ${[\I]}$); in particular, the $\I$'s are not coefficients of the $\SB{\pI}$'s specifying the S-basis, but rather coefficients for new basis vectors 
 $\IB{\pI}$, i.e.
\begin{align}\label{Ibasis3ext}
	\Sv{3} =  \sum_{\pI} \I_\pI \, \IB{\pI} \ ,
	\end{align}
where $\IB{\pI}$ denotes the I-basis vectors in the $\D$-dimensional entropy space.

Corresponding to the passive transformation \eqref{e:ItoS4} between the coefficients (or rather its inverse, which in this case is identical) that is expressible using a conversion matrix $\MStoI$
\begin{equation}
\S_{\pI} = \I_\pK \, \MStoI_{\pI}^{\ \pK}
\label{e:MStoIdef}
\end{equation}	
(with $\sum_\pK$ implied),
we can consider the active transformation on the basis vectors: 
\begin{align}
\IB{\pI} = \MStoI^{\ \ \pI}_{\pK} \, \SB{\pK} \ .
\end{align}
In other words, we convert the basis vectors with the transpose matrix of that converting the coefficients.
For example,  for $\N=2$ and $\N= 3$, we have the matrix $\MStoI$ in \eqref{e:MStoIdef}
\begin{equation}
\MStoI_{(\N=2)} =  
\begin{pmatrix}
1 & 0 & 0 \vspace{-0.2cm}  \\
0 & 1 & 0 \vspace{-0.2cm}  \\
1 & 1 & -1
\end{pmatrix}
\qquad {\rm and} \qquad
\MStoI_{(\N=3)} =  
\begin{pmatrix}
 1 \ & \ 0 \ & \ 0 \ & 0 & 0 & 0 & 0  \vspace{-0.2cm} \\
0 \ & \ 1 \ & \ 0 & 0 & 0 & 0 & 0 \vspace{-0.2cm} \\
0 & \ 0 \ & \ 1 & 0 & 0 & 0 & 0 \vspace{-0.2cm} \\
1 & \ 1 &\ 0 & -1 & 0 & 0 & 0 \vspace{-0.2cm} \\
1 & \ 0 & \ 1 & 0 & -1 & 0 & 0 \vspace{-0.2cm}  \\
0 & \ 1 & \ 1 & 0 & 0 & -1 & 0 \vspace{-0.2cm} \\
1 &\  1 &\  1 & -1 & -1 & -1 & 1 
\end{pmatrix}\ ,
\label{e:MStoI2}
\end{equation}
which  for $\N=2$ gives
$\IB{1}=\SB{1}+\SB{12}$, $\IB{2}=\SB{2}+\SB{12}$, and $\IB{12}=-\SB{12}$ when re-expressed in the S-basis.

At this stage, the S-basis and the I-basis appear rather symmetric, since 
$\MStoI = \MStoI^{-1}$.
However, compared to the polychromatic entanglement entropies themselves, the multipartite informations have a number of useful features, some of which were explored and utilized in \cite{Hubeny:2018ijt}.
In particular, as evident from \eqref{e:ItofromS}, the $n_\pI$-partite information $\I_\pI$ groups together $2^{n_\pI}-1$ polychromatic entanglement entropies (composed precisely out of all the monochromatic subsystems evoked by $\pI$) in a fully ($\mathbf{S}_{n_\pI}$) permutation symmetric fashion, whereas in the S-basis, this symmetry is upheld only by the (possibly smaller) combination of terms $\S_\pK$ with $\pK \supseteq \pI$. Furthermore, in the terminology of \cite{Hubeny:2018ijt}, the $\I_\pI$ are $(n_\pI-1)$-balanced, describing the increasingly tamer UV properties under increasing $n_\pI$.  In particular, for disjoint regions, every $\I_\pI$ is UV-finite for $n_\pI \ge 2$.  

Moreover, the $\I_\pI$'s have a number of interesting and useful recursion relations, which allow us to easily determine how these quantities change with the adding or subtracting of a system.  For example, if one of the arguments $\CA_i$ comprising our $\pI$ subsystem ($\{i\}\subset \pI$) is taken to vanish, i.e.\ $\CA_i = \emptyset$, so that we can take $\S_{\pK \cup \{i\}} = \S_\pK$ for any $\pK \nsupseteq \{i\}$, then each of the terms in $\I_\pI$  involving $\CA_i$ precisely cancels a corresponding one that does not include $\CA_i$, and therefore the multipartite information vanishes identically, i.e. $\I_\pI = 0$.

We can also consider what happens if we take the $(\N+1)$-party pure state and treat the purifier $\CA_{\N+1}$ on equal footing with the other $\CA_i$'s.  There is then a redundancy in how we write each term, since according to \eqref{e:Spurif}, for any $\pII$ (cf.\ footnote \ref{f:purifier_index}), we have $\S_\pII = \S_{\pII^c}$, where we defined the complement ${\pII^c} \equiv [\N+1] \setminus \pII$.
For $\N$ even, the purified expression then has opposite `parity', meaning terms with an odd number of monochromatic components purify to terms with an even number of components and vice versa, so that $(-1)^{n_\pII}= -(-1)^{n_{\pII^c}}$.   Since the multipartite information consists of terms where the signs alternate based on the number of components of the given term, all terms again cancel (and the last one vanishes by itself since $\S_{[\N+1]}=0$), so that
\begin{align}
\I_{\N+1}(\CA_i\!:\cdots:\!\CA_{\N+1}) =0 
\qquad {\rm for} \  \N\  {\rm even\ .}
\end{align}
On the other hand, when $\N$ is odd, the purified terms have the same parity as the original ones, so once we recast each term containing $\CA_{\N+1}$ into its purified form, the pairs of terms add constructively to yield
\begin{align}
\I_{\N+1}(\CA_1\!:\cdots:\!\CA_{\N+1}) = 2 \, \I_{\N}(\CA_1\!:\cdots:\!\CA_{\N}) 
\qquad {\rm for} \  \N\  {\rm odd\ .}
\end{align}
Therefore, for odd $\N$, not only is $\I_{\N}(\CA_1\!:\cdots:\!\CA_{\N})$ manifestly $\mathbf{S}_\N$ symmetric, but also $\mathbf{S}_{\N+1}$ symmetric  (i.e.\ permutation and purification symmetric).

The above observation suggests that the I-basis has partially nice properties under purifications:  the largest-rank multipartite information $\I_\N$ is purification symmetric for odd $\N$ though not for even $\N$.  This motivates us to devise a basis which will capture the purification symmetry for all $\N$, as well as for substructures with $n<\N$.  We will retain some structural distinction between even and odd $\N$, but in a subtler way that we describe next.

% _ _ _ _ _ _ _ _ _ _ _ _ _ _
\subsection{K-Basis}
\label{ss:Kbasis}
% - - - - - - - - - - - - - - - -

Let us now introduce our new basis, the K-basis.\footnote{\
We choose our naming to parallel the S-basis and I-basis, namely based on the labels of the coefficients, which in this case we denote by $\k$.
}
As motivated in the Introduction, we wish to treat the purifier $\CA_{\N+1}$ on equal footing with the monochromatic regions $\CA_{\N}$, and choose building blocks that respect this symmetry.  
In this context, we modify our notation accordingly, using $\ii=1,\ldots,\N+1$ and $\pII \in \powerset([\N+1]) \setminus \emptyset$ labeling any non-trivial subset of $[\N+1]$.
However, instead of allowing all possible subsystems $\pII$, we will only use a subset of them (specified below) and denote those by a distinct index $\gI$ to emphasize this fact.
As we shall see, although each building block only retains a subgroup of $\mathbf{S}_{\N+1}$, it does not single out the purifier in any way.  The key feature will be to use entropy structures attained by perfect tensors.

 A $2s$-perfect tensor, which we denote by \PTt{2s}, is a $2s$-party pure state for any positive integer $s$ such that the reduced density matrix involving any $s$ parties is maximally mixed.
 Holographically, using the RT formula \cite{Ryu:2006bv}, one may envision a \PTt{2s} as a wormhole with $2s$ boundaries, each with a bottleneck of area $\frac{1}{4G_N}$. In the simple case where $s=1$ and we have a wormhole between $\CA_1$ and $\CA_2$, this is just a Bell pair with entanglement entropies\footnote{\ 
We are ignoring factors of $\log 2$, and adopt the convention that a Bell pair has one unit of entanglement.}
\begin{align}
	\S_1 = \S_2 = 1\ , \quad \S_{12}=0 \ .
\end{align}
When such a structure is embedded in a larger system, $\N>1$, the $\CA_1\CA_2$ subsystem is completely uncorrelated with the rest, so we have additionally, for any $\pII$ that does not include $\CA_1$ or $\CA_2$,
\begin{align}
	\S_{1\pII}= \S_{2\pII} = 1\ , \quad \S_{12\pII}= \S_{\pII}=0 \ .
\end{align}
For example, for $\N=3$, a $\{12\}$ PT (cf.\ left diagram of Fig.~\ref{f:PTfig}) would have the entropy vector in the S-basis (cf.\ \eqref{Sbasis3}) 
\begin{align}
	\SPT{12} = \{1,1,0,0,1,1,0\} \ .
\end{align}

More generally, we may define any \PTt{2s} involving the parties $\CA_1,\ldots,\CA_{2s}$, where $2s \leq \N+1$, to have its entropy vector $\SPT{1\cdots 2s}$  given by
\begin{align}\label{e:SforK}
	\S_{{i_1}\cdots {i_t}} = \S_{{i_1}\cdots {i_t} \, \pII} = 
	\min \{ t , 2s-t \}\ ,
\end{align}
with  $i_1,\ldots,i_t \in \{1,\ldots,2s\}$ 
and $\pII$ lying in the complement, $i_1,\ldots,i_t \nsubseteq \pII$ (note that this expression also extends to $t=0$ where we take the left-hand-side to be simply $\S_\pII$, which vanishes); see Fig.~\ref{f:PTfig} for some examples of PTs.   
The data given in \eqref{e:SforK} is enough to specify the entropy vector completely.  Moreover, since a \PTt{2s} involving parties $\G = \{1,\ldots,2s\}$ is a $2s$-party pure state,  we can use purification symmetry within $\G$ itself to obtain $\S(\bar \CX) = \S(\CX)$ (already implied by \eqref{e:SforK}), where $\CX \subseteq \Gamma$ is some subset of the $2s$ parties and $\bar \CX$ denotes the complement of $\CX$ in $\G$.\footnote{\
For convenience, we denote this restricted complement by an overbar, 
$\bar \CX \equiv \G \setminus \CX$, 
to distinguish it from the complement in the full $[\N+1]$, which is indicated by ${\CX}^c \equiv [\N+1] \setminus  \CX$.
} 

%%%%%%%%%%%%%%%%%%%%%%%%%%%%%%%%
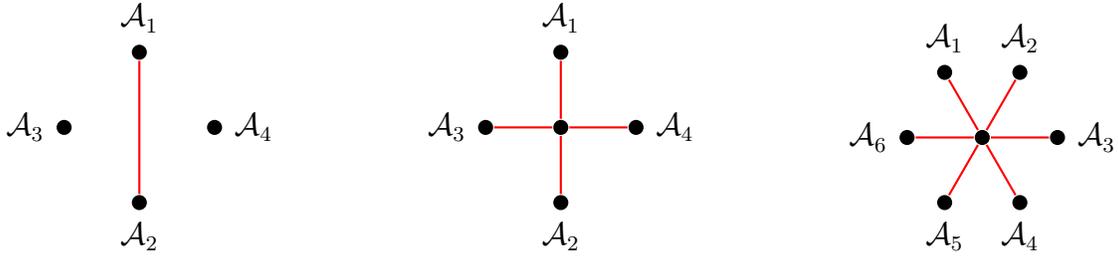
\begin{figure}
\centering
    \begin{tikzpicture}[
       thick,
       acteur/.style={
         circle,
         fill=black,
         thick,
         inner sep=2pt,
         minimum size=0.2cm
       }
     ] 
	\node (a1) at (0,2) [acteur,label=$\CA_1$]{}; 
      \node (a2) at (0,0) [acteur,label=below:$\CA_2$]{}; 
      \node (a3) at (-1,1) [acteur,label=left:$\CA_3$]{}; 
      \node (a4) at (1,1) [acteur,label=right:$\CA_4$]{};
      
      \draw[red] (a1) -- (a2);
\end{tikzpicture} \qquad\qquad     \begin{tikzpicture}[
       thick,
       acteur/.style={
         circle,
         fill=black,
         thick,
         inner sep=2pt,
         minimum size=0.2cm
       }
     ] 
	\node (a1) at (0,2) [acteur,label=$\CA_1$]{}; 
      \node (a2) at (0,0) [acteur,label=below:$\CA_2$]{}; 
      \node (a3) at (-1,1) [acteur,label=left:$\CA_3$]{}; 
      \node (a4) at (1,1) [acteur,label=right:$\CA_4$]{};
      \node (a5) at (0,1) [acteur,label={}] {};
      
       \draw[red] (a1) -- (a5);
       \draw[red] (a2) -- (a5);
       \draw[red] (a3) -- (a5);
      	\draw[red] (a4) -- (a5);
\end{tikzpicture} \qquad\qquad \begin{tikzpicture}[
       thick,
       acteur/.style={
         circle,
         fill=black,
         thick,
         inner sep=2pt,
         minimum size=0.2cm
       }
     ] 
	\node (a1) at (-1/2,{1+sqrt(3)/2}) [acteur,label=$\CA_1$]{}; 
      \node (a2) at (1/2,{1+sqrt(3)/2}) [acteur,label=:$\CA_2$]{}; 
      \node (a3) at (1,1) [acteur,label=right:$\CA_3$]{}; 
      \node (a4) at (1/2,{1-sqrt(3)/2}) [acteur,label=below:$\CA_4$]{};
      \node (a5) at (-1/2,{1-sqrt(3)/2}) [acteur,label=below:$\CA_5$]{}; 
      \node (a6) at (-1,1) [acteur,label=left:$\CA_6$]{}; 
      \node (a7) at (0,1) [acteur,label={}] {};
      
       \draw[red] (a1) -- (a7);
       \draw[red] (a2) -- (a7);
       \draw[red] (a3) -- (a7);
      	\draw[red] (a4) -- (a7);
      	\draw[red] (a5) -- (a7);
      	\draw[red] (a6) -- (a7);
\end{tikzpicture}
\caption{
We represent the entropy structure of various PTs using discrete graphs. 
The boundary vertices denote the monochromatic parties, and the set of all boundary vertices in each graph is denoted by $\G$.
Each edge represents one unit of entanglement, and the entropy of $\CX \subseteq\G$ is the minimum number of edges crossing from $\CX$ to $\G\setminus\CX$. \textbf{Left:} This is the Bell pair entanglement between $\CA_1$ and $\CA_2$ in a 3-party system ($\CA_4$ is the purifier). \textbf{Middle:} This is the graph representation for the \PTt{4} in a 3-party system. \textbf{Right:} This is the graph representation for the \PTt{6} in a 5-party system ($\CA_6$ is the purifier).  
}
\label{f:PTfig}
\end{figure}
%%%%%%%%%%%%%%%%%%%%%%%%%%%%%

We now claim that the PT structures form a basis for the entropy space. The proof is relegated to Section~\ref{positivity}, but observe that given our $\N$-party system, there are $\binom{\N+1}{2s}$ possible \PTt{2s}'s. It follows the total number of possible PTs is 
\begin{align}
	\sum_{s=1}^{\lceil \N/2 \rceil} \binom{\N+1}{2s} = 2^{\N}-1\ .
\end{align}
Since this matches the dimension $\D$ of the entropy space, the PTs form a basis for the entropy space if they are all linearly independent. This turns out to be the case (essentially because each PT embodies a distinct structure of entanglement), and we call this basis the K-basis.\footnote{\label{f:PTpositivity}
One may wonder whether it is always possible to decompose the entropy vector of an arbitrary holographic state as a {\it positive} sum of the entropy vectors of various PTs. If this were possible, it would then be suggestive that the PTs form the fundamental building blocks for the underlying entanglement structure of \emph{any} holographic state. Unfortunately, this is in fact not always possible, since as we explain below, for large enough $\N$ such a linear combination can involve negative coefficients.  An explicit example is the extreme ray given by the five-boundary wormhole in Fig.~2 of \cite{Bao:2015bfa}.} 

The basis vectors given by the PTs are labeled by $\KB{\gI} \equiv \SPT{\gI}$, where the collective index $\gI$, analogous to $\pI$ for the S and I bases, is built out of even-numbered polychromatic subsystems $\pII$, with $n_\gI=2s$ for a \PTt{2s}.
For example, we can capture the information of the $\N=3$ entropy vector in the K-basis (cf.\ \eqref{Sbasis} for S-basis and \eqref{Ibasis3} for I-basis) via
\begin{align}\label{k3}
	\Sv{3}= \left\{\k_{12}^{(3)},\; \k_{13}^{(3)},\; \k_{14}^{(3)},\; \k_{23}^{(3)},\; \k_{24}^{(3)},\; \k_{34}^{(3)},\; \k_{1234}^{(3)}\right\}_{[\k]} \ ,
\end{align}
which consists of $\binom{4}{2} = 6$ Bell pairs (\PTt{2}) between the parties $\CA_1,\ldots,\CA_4$ as well as a \PTt{4} involving all four parties.
Written explicitly, we have 
\begin{align}
\label{k3explicit}
	\Sv{3}  = \k^{(3)}_{12} \, \KB{12} + \k^{(3)}_{13} \, \KB{13} + \k^{(3)}_{14} \, \KB{14} + \k^{(3)}_{23} \, \KB{23}  + \k^{(3)}_{24} \, \KB{24} + \k^{(3)}_{34} \, \KB{34} + \k^{(3)}_{1234} \, \KB{1234}\ .
\end{align}
Similarly, a general entropy vector in the K-basis would then be expressed as
\begin{align}
\label{Kbasis}
	\Sv{\N} = \sum_\gI \k_\gI^\n \, \KB{\gI} \ .
\end{align}

As in \sref{ss:Ibasis}, we can convert between the $\S$'s and $\k$'s, as well as between the basis vectors $\SB{\pI}$ and $\KB{\gI}$ via a conversion matrix 
$\MStoK$ and its transpose.  
In particular, 
\begin{equation}
\S_{\pI} = \k_\gI^\n \,  \MStoK_{\pI}^{\ \gI} 
\qquad \Longleftrightarrow \qquad
\KB{\gI} =  \MStoK^{\ \gI}_{\pI} \, \SB{\pI} \ .
\label{e:StoK}
\end{equation}	
For example, for $\N=3$, by equating the entropy vector \eqref{k3explicit} with \eqref{Sbasis3ext} and using \eqref{e:SforK}, we obtain the expressions\footnote{\ 
Notice that \eqref{entropy3} holds only for $\N=3$; for other $\N$, $\S_\gI$ is a different linear combination of $\k_\gI^{\n}$ since there are further PT structures that contribute to the entanglement entropy of the given polychromatic system. This motivates using the superscript $\n$ in $\k_\gI^\n$.}
\begin{align}
\label{entropy3}
\begin{split}
	\S_i &= \sum_{\jj \not= i}^4 \k_{i\jj}^{(3)} + \k^{(3)}_{1234}\ , \quad \S_{ij} = \sum_{\kk \not= i,j}^4 \left(\k^{(3)}_{i\kk} + \k^{(3)}_{j\kk}\right) + 2\k^{(3)}_{1234}\ . \\
\end{split}
\end{align}
These two relations are sufficient to determine all components of $\vec \S^{\,\3}$.
In particular,  for $\N=2$ and $\N=3$, they imply
\begin{equation}
\MStoK_{(\N=2)} =  
\begin{pmatrix}
1 & 1 & 0 \vspace{-0.2cm}  \\
1 & 0 & 1 \vspace{-0.2cm}  \\
0 & 1 & 1
\end{pmatrix}
\qquad {\rm and} \qquad
\MStoK_{(\N=3)} =  
\begin{pmatrix}
1 & 1 & 1 & 0 & 0 & 0 & 1  \vspace{-0.2cm} \\
1 & 0 & 0 & 1 & 1 & 0 & 1 \vspace{-0.2cm} \\
0 & 1 & 0 & 1 & 0 & 1 & 1 \vspace{-0.2cm} \\
0 & 1 & 1 & 1 & 1 & 0 & 2 \vspace{-0.2cm} \\
1 & 0 & 1 & 1 & 0 & 1 & 2 \vspace{-0.2cm}  \\
1 & 1 & 0 & 0 & 1 & 1 & 2 \vspace{-0.2cm} \\
0 & 0 & 1 & 0 & 1 & 1 & 1 
\end{pmatrix}\ .
\label{e:MStoI2}
\end{equation}
We can also convert directly between the K-basis and the I-basis (which is in some sense easier).  For example, for $\N=3$, this yields
\begin{align}
\begin{split}
	\k_{ij}^{(3)} = \frac{1}{2} \, \I_{ij}\ , 
	\quad \k_{1234}^{(3)} = -\frac{1}{2} \, \I_{123} \ ,
\end{split}
\end{align}
though the conversion is more complicated for larger $\N$.  We will utilize some of this structure in \sref{positivity}.  In particular, it is easy to see that for any \PTt{2s} denoted by $\gI$, which can be represented by a star graph with $n_\gI=2s$ legs, the multipartite information vanishes for any subsystem that includes at most half of the system $\gI$, 
\begin{equation}\label{eq:I_PT_half_vanish}
\I_\pII =0 \qquad \forall \ n_\pII \le s \ .
\end{equation}	
Moreover, when $n_\pII = n_\gI$, the only non-vanishing  $\I_\pII$ is precisely the one evoking the given PT, $\pII = \gI$.  For example, in the left panel of Fig.~\ref{f:PTfig}, $\I_{12}=2$ while all other $\I_{\ii \jj}=0$.

Having considered all three bases, let us sharpen our geometrical intuition for their interplay by considering the simplest case, namely $\N=2$, where the entropy space is 3-dimensional, and hence easy to visualize.
% Figure 
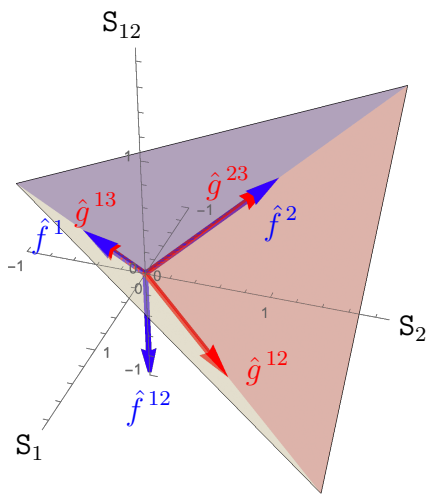
\begin{figure}[htbp]
\begin{center}
\includegraphics[width=3in]{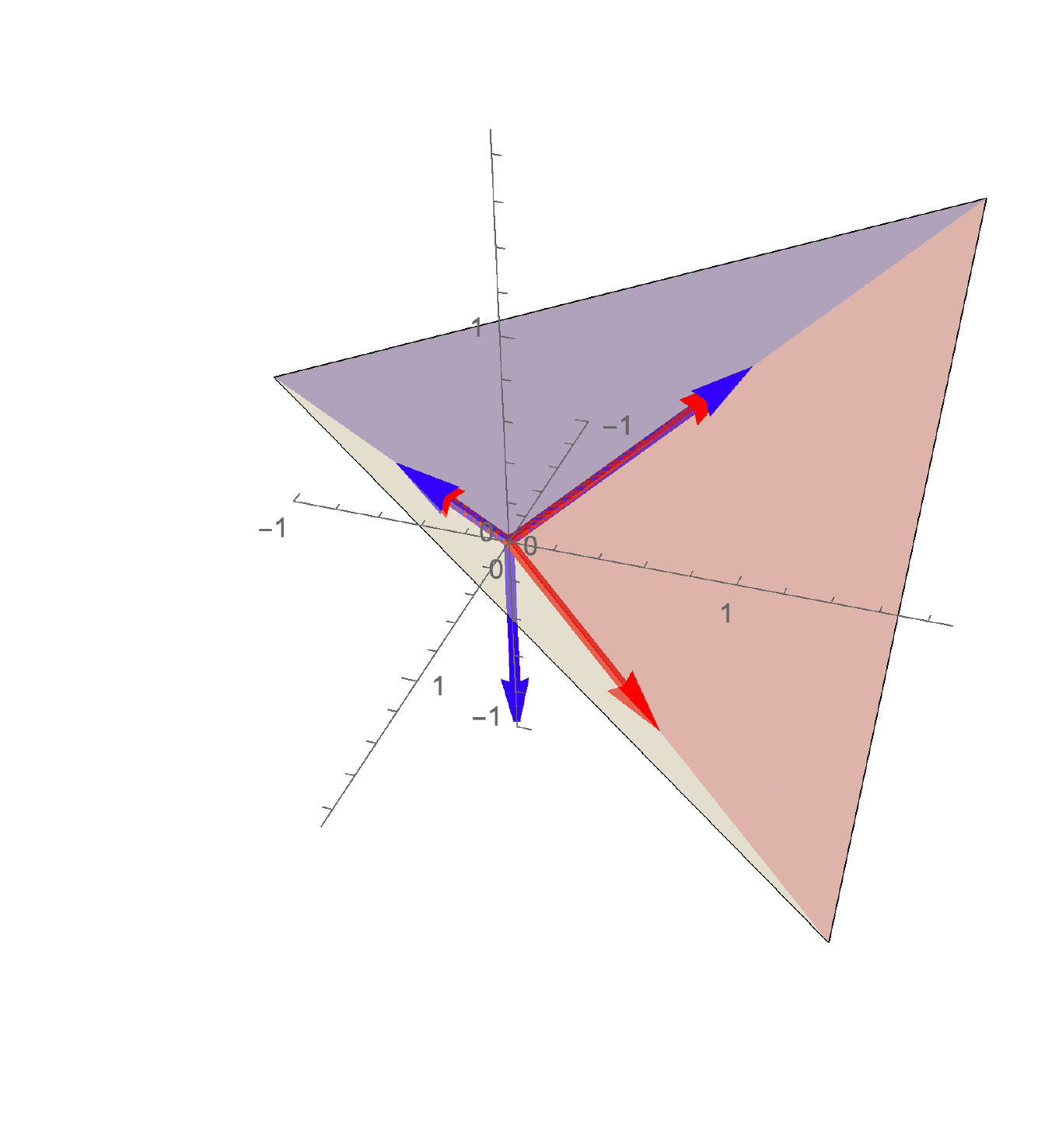}
\put(-165,50) {\makebox(20,20) {{$\S_1$}}}
\put(-20,95) {\makebox(20,20) {{$\S_2$}}}
\put(-130,210) {\makebox(20,20) {{$\S_{12}$}}}
\put(-158,130) {\makebox(20,20) {{\color{blue}{$\IB{1}$}}}}
\put(-70,135) {\makebox(20,20) {{\color{blue}{$\IB{2}$}}}}
\put(-120,65) {\makebox(20,20) {{\color{blue}{$\IB{12}$}}}}
\put(-75,80) {\makebox(20,20) {{\color{red}{$\KB{12}$}}}}
\put(-140,140) {\makebox(20,20) {{\color{red}{$\KB{13}$}}}}
\put(-90,150) {\makebox(20,20) {{\color{red}{$\KB{23}$}}}}
\caption{Entropy space for $\N=2$ with axes given by the 3 $\S_\pI$'s as labeled.  (The corresponding basis vectors $\SB{\pI}$, given by unit vectors pointing along the $\S_\pI$ axes, are not drawn to avoid clutter.)  The axes for the I-basis $\IB{\pI}$  (blue arrows) and K-basis $\KB{\G}$ (red arrows) are drawn, as well as the entropy cone delimited by the shaded facets (and defined by the extreme rays, which coincide with the K-basis vectors).}
\label{f:N2}
\end{center}
\end{figure}
This is illustrated in Fig.~\ref{f:N2}.  Two pairs of $\IB{\pI}$ and $\KB{\gI}$ in fact coincide, although the remaining pair does not.  Moreover, the K-basis vectors $\KB{\gI}$ are all proportional to the extreme rays of the holographic entropy cone (enclosed by the shaded facets in  Fig.~\ref{f:N2}).  For larger $\N$, we will prove in Theorem~\ref{extremalthm} of Section~\ref{positivity} that the basis vectors $\KB{\gI}$ are still proportional to the extreme rays, but there are additional extreme rays which are given by multi-term linear combinations of several $\KB{\gI}$'s.  (It is for this reason that not every entropy vector in the entropy cone can be expressed as a positive combination of the perfect tensors, as noted in footnote \ref{f:PTpositivity}.)

%_________________________________________
\section{Information Quantities}
\label{s:infoQ}
%-------------------------------------------------------

An information quantity  $\Q$ is given by a linear combination (with integer coefficients) of subsystem entropies $S_\pI$.  To each information quantity we can associate a hyperplane $\Q=0$ in the entropy space.  For sign-definite information quantities, we choose the overall sign such that $\Q\ge0$, obtaining entropy inequalities like SA or MMI.  Each such inequality restricts the physical states to lie in the corresponding half-space, and the full set of inequalities then delineates the holographic entropy polyhedron.  Here we are interested in primitive information quantities (cf.\ footnote \ref{f:primitive}), and particularly the sign-definite ones corresponding to holographic entropy inequalities.

An information quantity can be expressed in any of our three bases of interest: 
\begin{equation}
\Q
=\sum_{\pI }\mu_{\pI}\,\S_{\pI}
=\sum_{\pI }\nu_{\pI}\,\I_{\pI}
=\sum_{\gI}\lambda_{\gI}\, \k_{\gI} \ ,
\label{eq:}
\end{equation}	
so $\Q$ is equivalently specified by the coefficients $\mu_{\pI}$ in the S-basis, or $\nu_{\pI}$ in the I-basis, or $\lambda_{\gI}$ in the K-basis. 
As we shall see via Theorem~\ref{positivethm} in the next section, in the case of sign-definite information quantities, the coefficients $\lambda_{\gI}$ obey a positivity condition that is not shared by the $\mu_{\pI}$ or $\nu_{\pI}$ coefficients.

Let us start by examining the $\N=2$ and $\N=3$ entropy inequalities mentioned in the Introduction, along with uplifts of the $\N=2$ ones to $\N=3$.
For example, we now have two distinct versions of SA.  Direct uplift of \eqref{e:SA} (denoted as SA$_{(1,1)}$ because the left-hand-side involves two monochromatic subsystems) yields $\S_1+\S_2-\S_{12}\ge0$, which when rendered into the I-basis yields $\I_{12}\ge0$, and similarly when recast in the K-basis yields  $\k_{12}^{(2)}\ge 0$  for $\N=2$ and $\k_{12}^{(3)}\ge 0$  for $\N=3$.  We see that the I and K bases renditions are intrinsically simpler since they involve just a single term, in contrast to three terms in the S-basis.  The second possibility  (denoted as SA$_{(1,2)}$ because the left-hand-side involves a monochromatic and a 2-party polychromatic subsystem) is $\S_1+\S_{23}-\S_{123}\ge0$, which now has a more complicated rendering in the I-basis, $\I_{12}+\I_{13}-\I_{123}\ge0$, as well as in the K-basis, $\k_{12}^{(3)}+\k_{13}^{(3)}+\k_{1234}^{(3)}\ge0$.  However, whereas SA$_{(1,1)}$ is primitive, SA$_{(1,2)}$ is not  \cite{Hubeny:2018trv}.

\begin{table}[htbp]
\begin{center}
\footnotesize
\begin{tabular}{| c | c |l |l |l | c |}
\hline
$\N$ & Relation & S-basis & \shadeI{I-basis} & \shadeK{K-basis} & Nature \\
\hline
\hline
2,3 & SA$_{(1,1)}$ &
  $\S_{1}+\S_{2}-\S_{12}$ & 
  \shadeI{$\I_{12}$} &
  \shadeK{$\k_{12}^{(\N)}$ }
  & P, b
\\  \hline 
3 & SA$_{(1,2)}$ &
  $\S_{1}+\S_{23}-\S_{123}$ & 
  \shadeI{$\I_{12}+\I_{13}-\I_{123}$} &
  \shadeK{$\k_{12}^{(3)}+\k_{13}^{(3)}+\k_{1234}^{(3)}$ }
  & b
\\  \hline 
2 & AL$_{(1,2)}$ &
  $\S_{1}+\S_{12}-\S_{2}$ & 
  \shadeI{$2 \I_{1} - \I_{12}$} &
  \shadeK{$\k_{13}^{(2)}$ }
  & P
\\  \hline 
3 & AL$_{(1,2)}$ &
  $\S_{1}+\S_{12}-\S_{2}$ & 
  \shadeI{$2 \I_{1} - \I_{12}$} &
  \shadeK{$\k_{13}^{(3)}+\k_{14}^{(3)}+\k_{1234}^{(3)}$ }
  & 
\\  \hline 
3 & AL$_{(1,3)}$ &
  $\S_{1}+\S_{123}-\S_{23}$ & 
  \shadeI{$2 \I_{1} - \I_{12} - \I_{13}+\I_{123}$} &
  \shadeK{$\k_{14}^{(3)}$ }
  & P
\\  \hline 
3 & AL$_{(2,3)}$ &
  $\S_{12}+\S_{123}-\S_{3}$ & 
 \shadeI{\makecell{$2 \I_{1} + 2 \I_{2} - 2\I_{12}- \I_{13}$ \\ $ - \I_{23}+\I_{123}$}} &
  \shadeK{$\k_{14}^{(3)}+\k_{24}^{(3)}+\k_{1234}^{(3)}$ }
  & 
\\  \hline 
3 & SSA$_{(2,2)}$ &
  $\S_{12}+\S_{23}-\S_{2}-\S_{123}$ & 
 \shadeI{$\I_{13} - \I_{123}$} &
  \shadeK{$\k_{13}^{(3)}+\k_{1234}^{(3)}$ }
  & b
\\  \hline 
3 & WM$_{(2,2)}$ &
  $\S_{12}+\S_{23}-\S_{1}-\S_{3}$ & 
  \shadeI{$ 2 \I_{2}-\I_{12} -\I_{23}$} &
  \shadeK{$\k_{24}^{(3)}+\k_{1234}^{(3)}$ }
  & 
\\  \hline 
3 & MMI$_{(1,1,1)}$ &
  \makecell{$-\S_{1}-\S_{2}-\S_{3}+\S_{12}$ \\ $+\S_{13}+\S_{23}-\S_{123}$} & 
  \shadeI{$-\I_{123}$} &
  \shadeK{$\k_{1234}^{(3)}$ }
  & P, b$_2$
\\ \hline
\end{tabular}
\end{center}
\caption{Information quantities for $\N=2$ and $\N=3$, in the S, I, and K bases.  The last column indicates the information quantities that are primitive by ``P'' and those which are balanced by ``b" (or  ``b$_2$" when 2-balanced).}
\label{tab:summaryN3}
\end{table}

Table~\ref{tab:summaryN3} summarizes all the information quantities (up to permutations) corresponding to the inequalities mentioned in the Introduction (and their various uplifts) for $\N=2$ and $\N= 3$.  The name of the relation refers to the inequality in question, with the subscript referring to the sizes of the polychromatic subsystems evoked by the terms with positive coefficients when written in S-basis.  
The last column denotes the nature of the quantity, in particular whether it is balanced (indicated by ``b'') or 2-balanced (indicated by ``b$_2$'', cf.\ footnote \ref{f:Rbalance}), and whether it is primitive (``P'', cf.\ footnote \ref{f:primitive}).

Observe that Table \ref{tab:summaryN3} manifests several salient features:
\begin{itemize}
\item Information quantities rendered in the K-basis are simpler (i.e.\ involve fewer terms) than when rendered in the S-basis.  
\item Primitive quantities are much simpler in the K-basis (for the $\N=2,3$ cases, they consist of just a single term) than non-primitive quantities (namely, ones which are given by the sum of multiple primitive quantities).

\item Most importantly, all the coefficients in the K-basis are non-negative (unlike those in the S and I-bases).
\end{itemize}

It turns out that these properties remain true for larger $\N$ as well.  However, since the number of interesting information quantities grows with $\N$, we relegate the explicit details to Appendix \ref{a:QinSIKbases}.  In particular,
Appendix \ref{a:QinSIKbases} contains analogous tables for the  sign-indefinite primitive information quantities for $\N=4$ (Table \ref{tab:summaryN4}) and the sign-definite primitive information quantities for $\N=5$ (Table \ref{tab:summaryN5}), both up to permutations, and the latter up to purifications.  We have also checked explicitly that the above properties hold for all sign-definite information quantities for $\N=4$ (which are all uplifts of the lower $\N$ inequalities) as well as hitherto-known sign-definite $\N=6$ information quantities.\footnote{\
We thank Sergio Hernandez Cuenca for sharing the $\N=6$ information quantities with us.}

To distill the information indicating the `complexity' of the expression (defined as the number of distinct terms, but ignoring the value or sign of the nonzero coefficients), we consolidate this abreviated information in Table~\ref{termnumberN4} and Table~\ref{termnumberN5}.  In particular,
the two tables summarize the number of terms comprising these expressions in the three bases, giving a rough indicator of which basis is most convenient.\footnote{\ 
Note that while the number of terms is invariant under color permutation in each basis, it is additionally invariant under purifications only in the S and K basis (since the K-basis treats the purifier on equal footing and the S-basis just swaps terms) but not in the I-basis, so the number of terms for the I-basis may change for different instances of the same inequality.  We already saw a manifestation of this in Table \ref{tab:summaryN3}, and will encounter further examples below.
} 
We now comment on the results in a bit more detail.
\begin{table}[htbp]
\begin{center}
\footnotesize
\begin{tabular}{|l|c|c|c|c|c|c|c|}
\hline
Basis & $\I_4$ & $Q_1^{(4)}$ & $Q_2^{(4)}$ & $Q_4^{(4)}$ & $Q_5^{(4)}$ & $Q_6^{(4)}$ & $Q_7^{(4)}$ \\
\hline
S & 15 & 8 & 8 & 9 & 9 & 11 & 11 \\ \hline
\shadeI I & \shadeI 1 & \shadeI 2 & \shadeI 2 & \shadeI 3 & \shadeI 4 & \shadeI 4 & \shadeI 5  \\ \hline
\shadeK K & \shadeK 1 & \shadeK 2 & \shadeK 2 & \shadeK 4 & \shadeK 4 & \shadeK 5 & \shadeK 5  \\ 
\hline
\end{tabular}
\end{center}
\caption{Sign-indefinite primitive information quantities for $\N=4$ (the subscripts on the $Q_i^{(4)}$'s adhere to the conventions of \cite{Hubeny:2018ijt}; $Q_3^{(4)}$ is absent because it is an uplift of MMI and thus sign-definite, and the remaining $Q_i^{(4)}$'s are pairwise related by purification; see \cite{Hubeny:2018ijt} for more details). The rows denote the number of terms in each expression when expressed in S, I, or K basis. For a more comprehensive table specifying the information quantities explicitly, see Table~\ref{tab:summaryN4}  in Appendix~\ref{a:QinSIKbases}.}\label{termnumberN4}
\end{table}

Table \ref{tab:summaryN4} lists all sign-indefinite primitive information quantities for $\N=4$ (which have already been written down in  \cite{Hubeny:2018ijt}  in S and I bases.)  Rather satisfyingly, the signs of the $\lambda_\gI$ coefficients are mixed in all cases (except the most trivial one, corresponding to 4-partite information).   The I-basis version manifests that all  $Q^{(4)}_j$'s are 2-balanced, while $\I_4$ is 3-balanced.   The number of terms in each expression is summarized in Table \ref{termnumberN4}.  We see that in the S-basis, our expressions in each case involve at least twice as many terms as in the I or K bases.  
On the other hand, the expressions are comparably simple in the I and K bases, and in fact there are instances in the full $\mathbf{S}_{\N+1}$ symmetry orbit where the I-basis expressions have one fewer term than the K-basis ones.

Table \ref{tab:summaryN5} lists all sign-definite primitive information quantities for $\N=5$ (which have already been summarized in  \cite{Cuenca:2019uzx} in the S-basis, and originally found in  \cite{Bao:2015bfa}).   To give an example, the simplest new $\N=5$ inequality is the so-called 5-party cyclic inequality. In the S-basis, it is usually written in the manifestly cyclically symmetric form
\begin{align}\label{cyclicS}
\begin{split}
	\S_{123} &+ \S_{234} + \S_{345} + \S_{145} + \S_{125}  - \S_{12} - \S_{23} - \S_{34}  - \S_{45} - \S_{15} - \S_{12345} \geq 0\ ,
\end{split}
\end{align}
which involves 11 terms.  Re-expressed in the I-basis, this becomes
\small{
\begin{align}\label{cyclicI}
-\I_{124} - \I_{134} - \I_{135} - \I_{235} - \I_{245} + \I_{1234} + \I_{1235} + \I_{1245} + \I_{1345} + \I_{2345} - \I_{12345}\geq 0\ ,
\end{align}
}
which likewise involves 11 terms.\footnote{\ 
In this case, we can in fact do much better under purifications.  For instance, purifying with respect to $\CA_1$ yields $-\I_{124} - \I_{134} - \I_{135}+ \I_{1234} +\I_{1345} $, which has only five terms.
} However, in the K-basis, it is
\begin{align}\label{cyclicK}
	2 \k_{123456}^{(5)} + \k_{1246}^{(5)} + \k_{1346}^{(5)} + \k_{1356}^{(5)} + \k_{2356}^{(5)} + \k_{2456}^{(5)} \geq 0\ ,
\end{align}
which only involves 6 terms.  
\begin{table}[htbp]
\begin{center}
\footnotesize
\begin{tabular}{|l|c|c|c|c|c|c|c|c|}
\hline
Basis & SA$_{(1,1)}$ & MMI$_{(1,1,1)}$ & MMI$_{(1,2,2)}$ & $Q_1^{(5)}$ & $Q_2^{(5)}$ & $Q_3^{(5)}$ & $Q_4^{(5)}$ & $Q_5^{(5)}$ \\
\hline
S & 3 & 7 & 7 & 11 & 16 & 19 & 16 & 22 \\ \hline
\shadeI I & \shadeI 1 & \shadeI 1 & \shadeI 9 & \shadeI{11} & \shadeI 7 & \shadeI 7 & \shadeI 6 & \shadeI{10} \\ \hline
\shadeK K & \shadeK 1 & \shadeK 3 & \shadeK 5 & \shadeK 6 & \shadeK 7 & \shadeK 7 & \shadeK 8 & \shadeK{10} \\ 
\hline
\end{tabular}
\end{center}
\caption{Sign-definite primitive information quantities for $\N=5$. The rows denote the number of terms in each expression when expressed in S, I, or K basis. For a more comprehensive table, see Table~\ref{tab:summaryN5} in Appendix~\ref{a:QinSIKbases}.}\label{termnumberN5}
\end{table}
Table~\ref{termnumberN5} summarizes the number of terms involved in each of the other inequalities.  We see that the K-basis is still overall  the most compact of the three, though under some purifications the I-basis quantities can be rendered at least as compact as the K-basis.\footnote{\ 
Apart from the simplification for $Q_1^{(5)}$ mentioned above, purifying MMI$_{(1,2,2)}$ with respect to one of the elements in the polychromatic arguments, such as $\CA_2$, yields MMI$_{(1,1,2)}$, which involves only three terms in the I-basis.
} 
As previously, in all cases the S-basis expressions contain many more terms. 

Part of the reason why many of the holographic entropy inequalities simplify in the K-basis (at least for $\N \leq 5$) stems from the fact that the basis elements of the K-basis are proportional to the extreme rays in the holographic entropy cone (see Theorem~\ref{extremalthm} in the next section), which have a privileged role.  The hyperplanes forming the facets of the cone are generated by the span of $\D-1$ extreme rays.  The normal vector to each hyperplane has components that give the coefficients in the corresponding inequalities, so by construction the given inequality cannot evoke any of the $\k^{\n}_\gI$'s that comprise the generating extreme vectors.  This therefore decreases the eligible set of $\k^{\n}_\gI$'s that can be evoked in the inequality.  

Moreover, the observation that 
every inequality can be written in the form
\begin{align}\label{conj}
	\sum_{\gI} \l_\gI \,  \k^{\n}_\gI \geq 0
\end{align}
with non-negative coefficients $\l_\gI$ explains why the primitive information quantities are generally simpler than the non-primitive ones.
For the non-primitive inequalities, which are defined to be sums of the primitive inequalities, the non-negative nature of $\lambda_\gI$ in \eqref{conj} implies that cancellation between terms in different primitive inequalities cannot occur, and the number of terms can only grow when we sum two inequalities. In the next section, we prove both the positivity of $\lambda_\gI$, as well as the fact that the K-basis elements are all proportional to extreme rays.

\section{K-Basis Properties}\label{positivity}

As a warm up, we begin by first proving that the K-basis is indeed a basis.\footnote{\ 
Of course, this is implicitly obvious from our construction, which provides a surjection between $\S_\pI$'s and $\k_\gI$'s, but it will be instructive in developing our toolkit to demonstrate this explicitly.
} We then turn to proving some nice properties of the K-basis  in the following two subsections.

\begin{theorem} The elements $\hat g^\G$ corresponding to all \PTt{2s} with $2s\le \N+1$ are linearly independent and span the entropy space. 
\end{theorem}

\begin{proof}

Linear independence and completeness are enough to show that $\{\hat g^\G\}$ form a basis, which we call the K-basis.  
Completeness is easy:
Recall that in an $\N$-party system, there are $\binom{\N+1}{2s}$ possible PT$_{2s}$'s, so the total number of PTs is
	\begin{align}
	\sum_{s=1}^{\lceil \N/2 \rceil} \binom{\N+1}{2s} = 2^{\N}-1 = \D\ ,
\end{align}
which matches the total dimension of the entropy space.   We therefore only need to prove that the entropy vectors associated to the PTs are all linearly independent.

We proceed by contradiction. Suppose that $\big\{\KB{\gI}\big\}$ are in fact linearly dependent.  This means that there is some subset of $\{\gI\}$, call it $\{\gInz\}$, for which 
\begin{align}\label{linindep}
	\sum_{\gInz} \a_\gInz \, \hat g^\gInz = 0\ ,
	 \qquad {\rm with} \  \ \ \a_\gInz \ne 0  \ \ \forall\ \gInz \ .
\end{align}

We choose $\G_0 \in \{\gInz\}$ to be a subsystem that has the highest cardinality, namely $n_{\G_0} \geq n_{\tilde\G} \ \forall\ \gInz$. Then $\G_0$ is a PT$_{n}$ for some $n \equiv 2s$, and we may without loss of generality assume it involves the parties $\CA_1,\ldots,\CA_{n}$. We then compute the $n$-partite information $\I_{12\cdots n}$ for $\G_0$ to be
\begin{align}
\begin{split}
	\left.\I_{12\cdots n}\right|_{\big(\hat g^{\G_0}\big)} &= \sum_{k=1}^{n}(-1)^{k+1}\binom{n}{k}\min(k,n-k) \\
	&=\sum_{k=1}^{n/2}\left[(-1)^{k+1}\binom{n}{k}k + (-1)^{k+\frac{n}{2}+1}\binom{n}{k+\frac{n}{2}}\left(\frac{n}{2}-k\right)\right]\\
	&= \frac{(-1)^{\frac{n}{2}+1}\left(\frac{n}{2}+1\right)}{n-1}\binom{n}{\frac{n}{2}+1}\ ,
\end{split}
\end{align}
where the subscript $\big(\hat g^{\G_0}\big)$ indicates the entropy vector on which we are evaluating the multipartite information. In particular, this expression is nonzero.

On the other hand, consider now any other $\G'$ where $\a_{\G'} \not= 0$. Suppose the number of parties in $\G'$ is $n_{\G'} = m$ with $m \leq n$, and that $n_{\G' \cap \G_0} = p$ for some $p < n$ (note if $p=n$, then $\G' = \G$). This means there is some $i \in [n]$ such that $\CA_i \in \G_0$ but $\CA_i \notin \G'$. If we evaluate $\I_{12\cdots n}$ on $\hat g^{\G'}$, this amounts to setting $\CA_i$ to zero. By the argument given in Section~\ref{ss:Ibasis}, this means $\I_{12\cdots n}$ associated to $\hat g^{\G'}$ vanishes.\footnote{\ 
This fact can also be deduced via combinatorics explicitly. The $n$-partite information $\I_{12\cdots n}$ for $\G'$ is given by
\begin{align}
\begin{split}
	\left.\I_{12\cdots n}\right|_{\big(\hat g^{\G'}\big)} &= \sum_{k=1}^{n} (-1)^{k+1}\sum_{l=1}^k\binom{p}{l}\binom{n-p}{k-l}\min(l,m-l) \\
	&= \sum_{l=1}^p \binom{p}{l}\min(l,m-l)\sum_{k=l}^{n} (-1)^{k+1}\binom{n-p}{k-l} \\
	&=  \sum_{l=1}^p (-1)^{-l}\binom{p}{l}\min(l,m-l)\sum_{k=0}^{n-p} (-1)^{k+1}\binom{n-p}{k}\ .
\end{split}
\end{align}
 Because $p < n$, the second sum in the last line vanishes, so the $n$-partite information associated to $\hat g^{\G'}$ vanishes as well.} This is true for all $\G' \not= \G_0$ with nonzero $\a_{\G'}$, so in order for \eqref{linindep} to hold, we need $\a_{\G_0} = 0$, a contradiction.

\end{proof}

% _ _ _ _ _ _ _ _ _ _ _ _ _ _
\subsection{Positivity of the K-Basis}
% - - - - - - - - - - - - - - - -

The fact that all the known holographic entropy inequalities involve only positive linear combinations of $\k_\pI^\n$ (cf.\ e.g.\ Table \ref{tab:summaryN5} in Appendix \ref{a:QinSIKbases}) strongly suggests the following theorem. 
\begin{theorem}\label{positivethm} Given any $\N$-party system, the holographic entropy inequalities, when expressed in the K-basis, can all be written in the form
\begin{align}\label{thmeqn}
	\sum_{\G} \l_\G \, \k^{(\N)}_\G \geq 0  
	\qquad {\rm with} \ \l_\gI \geq 0 \ \ \forall \ \G \ ,
\end{align}
where $\gI$ labels the subsets of the parties (including the purifier) that have an even number of elements. 
\end{theorem}
As we shall see, once we realize the fact that the basis elements of the K-basis lie within the entropy cone, the proof is almost trivial.
\begin{proof}
	We begin by observing that every basis element $\hat g^\G$ in the K-basis is a vector within the entropy cone, given that all the PTs are realizable as multiboundary wormhole geometries in the bulk.\footnote{\label{f:metric_cone} 
While this claim was made and used in \cite{Bao:2015bfa}, the additional requirement that the bulk geometries used for its construction should be dual to physical states in a CFT render this argument inconclusive \cite{Marolf:2017shp}.  In particular, using earlier results of \cite{Maxfield:2016mwh}, the authors argue that the simplest such geometries are likely to correspond to  subdominant bulk phases of natural path integrals.  In the language of \cite{Marolf:2017shp}, the K-basis vectors are guaranteed to lie inside the {\it metric entropy cone} but not necessarily the {\it HRT cone} (which is the more appropriate definition of the holographic entropy cone).  Here we will simply adhere to the common assumption that there indeed exist CFT states realizing the K-basis vectors, and use it both in this subsection and the next (in proving Theorem~\ref{extremalthm}).
}
Since any vector in entropy space (not necessarily in the cone) is of the form
	\begin{align}
	\sum_\G \k_\G^\n \hat g^\G\ ,
\end{align}
this implies for any fixed $\G'$, setting $\k_{\G}^\n = \delta_{\G'\G}$ yields an entropy vector that lies within the holographic entropy cone. Since all the entropy vectors lying within the cone satisfy \eqref{thmeqn}, substituting $\k_{\G}^\n = \delta_{\G'\G}$ into \eqref{thmeqn} implies
\begin{align}
	\l_{\G'} \geq 0\ .
\end{align}
This is true for any fixed $\G'$, so the result follows.

\end{proof}

In Appendix~\ref{a:contraction}, we provide a weaker version of the proof that only applies to holographic entropy inequalities provable using the \emph{method of contraction} introduced in \cite{Bao:2015bfa}. This mainly serves to illustrate another usage of the contraction mapping, although because all known holographic entropy inequalities are provable by contraction, this may also be an alternative proof to Theorem~\ref{positivethm} if indeed all entropy inequalities fall into this category.

% _ _ _ _ _ _ _ _ _ _ _ _ _ _
\subsection{Relation to Extreme Rays}
% - - - - - - - - - - - - - - - -

The simplification of the holographic entropy inequalities for low values of $\N$ in the K-basis as opposed to the S-basis begs for an explanation. In this subsection, we provide partial intuition for why this may be the case with the following theorem.\begin{theorem} In the K-basis, the basis elements $\hat g^\G$ are proportional to extreme rays in the holographic entropy cone. \label{extremalthm}
\end{theorem}
\begin{proof}
	An extreme ray in the holographic entropy cone is proportional to a vector $\vec \S$ that can only be written as a positive linear combination of rays proportional to $\vec \S$ in the cone. Let $\G$ be a fixed PT$_{2s}$, and suppose that 
	\begin{align}\label{assumption}
	\hat g^\G  = \sum_{k=1}^K a_k \, \vec \S^{(k)}\ ,
\end{align}
where $\vec \S^{(k)}$ are some entropy vectors in the cone, and $a_k > 0$ are positive coefficients. 
We want to show that each  $\vec \S^{(k)}$ must then be proportional to $\hat g^\G$.
Without loss of generality, we may assume $\G$ is the PT involving parties $\CA_1,\ldots,\CA_{2s}$. Clearly the monochromatic entropies $\S^{(k)}_i$ must only involve the parties $\CA_1,\ldots,\CA_{2s}$; otherwise, if $\S^{(k)}_j > 0$ for $j \notin [2s]$, then the $j$-th component on the right-hand-side of \eqref{assumption} is positive, while the $j$-th component of $\hat g^\G$ vanishes, a contradiction.

We now observe that for any two disjoint subsets of parties $\CX_1$ and $\CX_2$ in $\G$, with cardinalities $n_{\CX_1}$ and $n_{\CX_2}$, respectively,
\begin{align}\label{zeromutual}
\begin{split}
	\left.\I(\CX_1:\CX_2)\right|_{\big( \hat g^\G\big)} = 0 \quad\text{iff}\quad n_{\CX_1} + n_{\CX_2} \leq s\ ,
\end{split}
\end{align}
where $\I(\CX_1:\CX_2)$ is shorthand for the mutual information between the union of the parties in $\CX_1$ with those in $\CX_2$, and the subscript $\big(\hat g^\G\big)$ indicates the entropy vector on which we are evaluating this mutual information.\footnote{\ 
Relation \eqref{zeromutual} is easy to see explicitly, but also follows trivially from \eqref{eq:I_PT_half_vanish} by expanding $\I(\CX_1:\CX_2)$ in the I-basis and noting that each term must vanish individually.
} 
Using \eqref{assumption} and the fact that mutual information is always non-negative (by SA), this means if $n_{\CX_1} + n_{\CX_2} \leq s$, $\I(\CX_1:\CX_2)$ must also vanish for every entropy vector on the right-hand-side of \eqref{assumption}, i.e.
\begin{align}
	\S^{(k)}_{\CX_1} + \S^{(k)}_{\CX_2} - \S^{(k)}_{\CX_1\CX_2} = 0 \quad\forall \ k \quad\text{if} \quad n_{\CX_1} + n_{\CX_2} \leq s\ .
\end{align}
This in particular means that for any subset $\CX \subset [2s]$ with $n_\CX \leq s$, we have
\begin{align}\label{bipartiteS}
	\S^{(k)}_{\CX} = \sum_{i \in \CX} \S^{(k)}_i \quad\forall \ k \ .
\end{align}

Consider now a fixed $\vec\S^{(k)}$. As it involves at most the parties $\CA_1,\ldots,\CA_{2s}$, we may label the parties such that the monochromatic entropies obey
\begin{align}\label{monochromaticS}
	\S^{(k)}_1 \leq \S^{(k)}_2 \leq \cdots \leq \S^{(k)}_{2s}\ .
\end{align}
We now want to examine the mutual information between $\pI_1 = \{1,\ldots,s\}$ and $\pI_2 = \{s+1,\ldots,2s\}$ associated to $\vec\S^{(k)}$.  Since $\vec\S^{(k)}$ is the entropy vector associated to a pure state involving $2s$ parties, 
\begin{align}
	\S^{(k)}_{\pI_1} = \S^{(k)}_{\pI_2}\ .
\end{align}
On the other hand, we know from \eqref{bipartiteS}
\begin{align}
	\S^{(k)}_{\pI_1} = \sum_{i=1}^{s} \S^{(k)}_i \leq \sum_{i=s+1}^{2s} \S^{(k)}_i = \S^{(k)}_{\pI_2}\ . 
\end{align}
Thus, there is no contradiction (i.e.\ the inequality above is actually an equality) if and only if all the inequalities in \eqref{monochromaticS} are equalities. This proves $\vec\S^{(k)}$ is actually proportional to $\hat g^\G$. As this is true for all $k$, it follows that $\hat g^\G$ can only be written as sums of vectors in the holographic entropy proportional to itself, thereby completing the proof.

\end{proof}

%_________________________________________
\section{Discussion}
\label{discussion}
%-------------------------------------------------------

We have examined the relative merits of three distinct bases for the $\N$-party entropy space (for arbitrary $\N$), namely the {\it S-basis} based on the entanglement entropies of polychromatic subsystems, the  {\it I-basis} based on multipartite information between the monochromatic subsystems, and the  {\it K-basis} based on even-party PTs.  While the S-basis is the most frequently used one \cite{Bao:2015bfa,Cuenca:2019uzx,Hubeny:2018trv,Hubeny:2018ijt}, and the I-basis was considered recently as well \cite{Hubeny:2018ijt}, the K-basis is a new construct.  This was motivated primarily by symmetry considerations: while the former two explicitly evoke quantities built only out of the $\N$ subsystems, the K-basis treats the purifier on equal footing, thereby manifesting the full permutation plus purification symmetry.

We have seen that primitive information quantities are rendered substantially more compactly in both the I-basis and the K-basis than in the S-basis.  While the comparison between the number of terms was less conclusive between the I and K bases (partly hindered by the fact that in the I-basis this quantity need not remain invariant under purifications), we have unearthed one major advantage of the K-basis:  Any information quantity must have non-negative coefficients when expressed in this basis (cf.\ Theorem~\ref{positivethm}).  This is rooted in the fact that the PTs not only form a basis, but also lie in the entropy cone, representing physically accessible states, which allowed us to extract the sign of the individual coefficients.  Any other basis composed solely of physically accessible vectors would share this positivity feature in representing entropy inequalities as well.

On the other hand, the astonishing reduction in the number of terms in primitive information quantities is a distinct feature,  more closely tied to the PTs' distinguishing property: all PTs are extreme rays of the entropy cone (cf.\ Theorem~\ref{extremalthm}), meaning they cannot be expressed as a positive combination of any other vectors in the cone.  Correspondingly, they saturate a large number of holographic entropy inequalities, such as subadditivity for any pair of subsystems whose combined cardinality does not exceed half of the PT's cardinality.  However, while for $\N=2,3$ all the extreme rays are conversely PTs, it is no longer true that all extreme rays are PTs for larger $\N$.  This is already evident from the fact that there are more extreme rays than the dimensionality $\D$ of entropy space, whereas there are only $\D$ (even party) PTs by construction.\footnote{\ 
The remaining extreme rays originate from intersections of hyperplanes corresponding to higher-rank inequalities, which there are many more of with increasing $\N$.  (For example, while the  $\N=4$ holographic entropy cone has 16 extreme rays, the $\N=5$ one already has 2,267 extreme rays \cite{Cuenca:2019uzx}.)}

Given the success of the K-basis, one obvious question is whether we can do even better with yet another basis (or perhaps some convenient over-complete representation).  Since the even-party PTs forming the K-basis have so many nice properties, one might naturally wonder about the odd-party PTs (defined again as being maximally mixed for any bipartition within the PT structure and represented analogously as an equal-weight star graph with odd number of legs).  However, here the counting is not as clean as for the even-party ones:  The total number of polychromatic subsystems including the purifier is $2^{\N+1}-1$, which separates into those with even cardinality (which there are $\D = 2^\N-1$ of), and the remaining odd cardinality ones (with $2^\N$ elements).  Moreover, the monochromatic subsystems (along with the purifier) seem rather artificial as PTs, so discounting these we only have $2^\N-\N-1$ odd-party PTs, which is too few to form a basis.\footnote{\ 
On the other hand, if we do include the 1-PTs, we have one too many vectors to form a basis.  For even $\N$, where we only have a single \PTt{\N+1}, we can exclude it while retaining a permutation plus purification symmetric collection of $\D$ elements; however, this does not appear to retain any obvious advantages.
}  Nevertheless, this collection still has a number of special properties.  In particular, since each PT saturates SA for any partition into subsystems that make up less than half of the system, this vector must lie on a face of the entropy cone.  However, these cannot be extreme rays, since we can express each odd PT as a (normalized) sum over the 1-lower even PTs. 

If instead of candidate basis vectors for the entropy space we consider natural quantities (composed of a specific collection of entropies for polychromatic subsystems) in terms of which to express a given information quantity, another natural option presents itself.  In particular, we can consider a collection based solely on mutual information, but now of any non-overlapping pair of polychromatic subsystems (including the purifier so as to retain the full symmetry), i.e.\ $\{ \S_\pII + \S_\pKK - \S_{\pII\cup\pKK} \}$ for every $\pII,\pKK$ with $\pII \cap \pKK = \emptyset$.  This collection is over-complete; for example we can express the tripartite information in terms of three different such collections, $\I_{123}=\I(\CA_1:\CA_2)+\I(\CA_1:\CA_3)-\I(\CA_1:\CA_{23})$ and the two distinct permutations thereof.  While using a redundant representation sounds like a disadvantage, it may in fact render the packaging more compact and be more intimately tied to the fundamental nature of the holographic information quantities.  This representation will be examined in the upcoming work \cite{HHRR:2019aa}, which examines the holographic entropy arrangement \cite{Hubeny:2018ijt} from the point of view of phase transitions in the HRT surfaces that give the collection of entanglement entropies. 

Let us conclude by remarking on the role of holography.  The I and K bases were constructed in a purely algebraic manner.  This means that we can rewrite {\it any} information quantity in these bases, regardless of its relevance to holography.  One would expect that since the structures involved (multipartite informations and PT structures, respectively) have a natural quantum information theoretic meaning, that these bases will remain convenient even outside of the context of holography wherein they were employed.  In particular, the K-basis vectors, by virtue of saturating a large set of subadditivities, will still lie on the boundary of the quantum entropy cone, not just the holographic one.  More importantly, the proof of Theorem~\ref{extremalthm} of the previous section carries through in this broader context, so in fact the K-basis vectors are proportional to extreme rays of the full quantum entropy cone.\footnote{
In fact, the subtlety for holographic entropy (HRT) cone \cite{Marolf:2017shp} mentioned in footnote \ref{f:metric_cone} doesn't arise here, since the metric entropy cone is contained within the full quantum entropy cone 
\cite{Hayden:2016cfa}. Moreover, constructing PTs from qu$d$it states, while hard, is always possible for sufficiently large dimension $d$ \cite{Helwig:2012nha}.  
Furthermore, by the Page curve \cite{Page:1993wv}, a random state on the tensor product of $n$ identical large Hilbert spaces realizes a PT, which is another way to see that the PTs are realized in the quantum entropy cone.  We thank Michael Walter for alerting us to this argument.
}  It would therefore be interesting to explore the utility of these representations in the fully general context, transcending holography.

%~~~~~~~~~~~~~~~~~~~~~~~~~~~~~~~~
\section*{Acknowledgements}

We would like to thank Ning Bao, Sergio Hernandez Cuenca, Xi Dong, Mukund Rangamani, Massimiliano Rota, and Michael Walter  for stimulating conversations.  T.H.\ and V.H.\ were supported by U.S. Department of Energy grant DE-SC0009999 and by funds from the University of California.  M.H.\ was supported by the Simons Foundation through the ``It from Qubit'' Simons Collaboration, and by the U.S. Department of Energy under grant DE-SC0009987. M.H.\ would like to thank the UC Davis Center for Quantum Mathematics and Physics for hospitality.

%~~~~~~~~~~~~~~~~~~~~~~~~~~~~~~~~
\appendix
\section{Summary of Information Quantities in S, I, and K bases}
\label{a:QinSIKbases}

Here we summarize all interesting information quantities in the entropy (S), multipartite information (I), and perfect tensor (K) bases.
Since we already presented the case of $\N=2$ and $\N=3$ in Table \ref{tab:summaryN3} in \sref{s:infoQ}, and there are no other (sign-indefinite) information quantities, we focus on  $\N=4$ and $\N=5$.

% ---------------------- (N=4)
\subsubsection*{$\N=4$:} 

\begin{table}[htbp]
\begin{center}
\scriptsize
\begin{tabular}{| c | c |l  | }
\hline
Relation & \makecell{Basis}  & Primitive Information Quantity  \\ 
\hline
\hline
 $\I_4$
 &  S (15) &  \makecell{$\S_{1}+\S_{2}+\S_{3}+\S_{4}-\S_{12}-\S_{13}-\S_{14}-\S_{23}-\S_{24}-\S_{34} +\S_{123}+\S_{124}+\S_{134} +\S_{234}-\S_{1234}$} \\ 
  &  \shadeI{I (1)} & \shadeI{$\I_{1234}$}  \\
&  \shadeK{K (1)} & \shadeK{$- 2\k_{1234}^{(4)}$ } 
\\  \hline 
$Q^{(4)}_1$
 &  S (8) &  $\S_{1}-\S_{2}-\S_{13}-\S_{14}+\S_{23}+\S_{24}+\S_{134}-\S_{234}$ \\ 
  &  \shadeI{I (2)} & \shadeI{$\I_{134}-\I_{234}$}  \\
&  \shadeK{K (2)} & \shadeK{$- \k_{1345}^{(4)}+ \k_{2345}^{(4)}$ } 
\\  \hline 
 $Q^{(4)}_2$
 &  S (8) &  $\S_{1}-\S_{12}-\S_{13}-\S_{14}+\S_{123}+\S_{124}+\S_{134}-\S_{1234}$ \\ 
&  \shadeI{I (2)} & \shadeI{$-\I_{234}+\I_{1234}$}  \\
 &  \shadeK{K (2)} & \shadeK{$- \k_{1234}^{(4)}+ \k_{2345}^{(4)}$ } 
\\  \hline  
$Q^{(4)}_4$
 &  S (9) &  $2\S_{1}+\S_{2}-2\S_{12}-\S_{13}-\S_{14}+\S_{34}+\S_{123}+\S_{124}-\S_{234}$ \\ 
 &  \shadeI{I (3)} & \shadeI{$\I_{123}+\I_{124}-\I_{234}$}  \\
 &  \shadeK{K (4)} & \shadeK{$- \k_{1234}^{(4)}- \k_{1235}^{(4)}- \k_{1245}^{(4)}+ \k_{2345}^{(4)}$ } 
\\  \hline 
 $Q^{(4)}_5$
 &  S (9) &  $\S_{1}+\S_{23}+\S_{24}+\S_{34}-\S_{123}-\S_{124}-\S_{134}-2\S_{234}+2\S_{1234}$ \\ 
 &  \shadeI{I (4)} & \shadeI{$\I_{123}+\I_{124}+\I_{134}-2\I_{1234}$}  \\
 &  \shadeK{K (4)} & \shadeK{$\k_{1234}^{(4)}- \k_{1235}^{(4)}- \k_{1245}^{(4)}- \k_{1345}^{(4)}$ } 
\\  \hline 
 $Q^{(4)}_6$
 &  S (11) &  $3\S_{1}-2\S_{12}-2\S_{13}-2\S_{14}+\S_{23}+\S_{24}+\S_{34}+\S_{123}+\S_{124}+\S_{134}-2\S_{234}$ \\ 
 &  \shadeI{I (4)} & \shadeI{$\I_{123}+\I_{124}+\I_{134}-2\I_{234}$}  \\
 &  \shadeK{K (5)} & \shadeK{$- \k_{1234}^{(4)}- \k_{1235}^{(4)}- \k_{1245}^{(4)}- \k_{1345}^{(4)}+2 \k_{2345}^{(4)}$ } 
\\  \hline 
 $Q^{(4)}_7$
 &  S (11) &  $\S_{12}+\S_{13}+\S_{14}+\S_{23}+\S_{24}+\S_{34}-2\S_{123}-2\S_{124}-2\S_{134}-2\S_{234}+3\S_{1234}$ \\ 
 &  \shadeI{I (5)} & \shadeI{$\I_{123}+\I_{124}+\I_{134}+\I_{234}-3\I_{1234}$}  \\
 &  \shadeK{K (5)} & \shadeK{$2\k_{1234}^{(4)}- \k_{1235}^{(4)}- \k_{1245}^{(4)}- \k_{1345}^{(4)}- \k_{2345}^{(4)}$  } 
\\  \hline 
 \end{tabular}
\end{center}
\caption{Sign-indefinite primitive information quantities for $\N=4$ in the S, I, and K bases.   The expressions in S and I bases (which have already been written down in  \cite{Hubeny:2018ijt}) have the same normalization, while the K-basis expressions are divided by overall factor of 2 for convenience.  The I-basis rendering manifests that all  $Q^{(4)}_j$s are 2-balanced, while $\I_4$ is 3-balanced.  The number of terms for each quantity is indicated in parentheses in the second column and is summarized in Table \ref{termnumberN4}. }
\label{tab:summaryN4}
\end{table}

For a 4-party system, the entropy vector has 15 components, which in the K-basis corresponds to the $\binom{5}{2} = 10$ Bell pairs and $\binom{5}{4} = 5$ PT$_4$. 
The generating expression for the conversion to the K-basis, analogous to \eqref{entropy3} becomes
\begin{align}\label{entropy4}
\begin{split}
	\S_1 &= \sum_{j = 2}^5 \k_{1j}^{(4)} + \sum_{1<j < k < l}^5 \k^{(4)}_{1jkl} \\
	\S_{12} &= \sum_{j =3}^5 \left( \k_{1j}^{(4)} + \k_{2j}^{(4)}\right) +  \k_{1345}^{(4)} + \k_{2345}^{(4)} + 2\sum_{2 < j < k}^5 \k^{(4)}_{12jk}\ .
\end{split}
\end{align}
Using the previous result of \cite{Hubeny:2018ijt}, which found all primitive $\N=4$ information quantities (and presented them in the S and I basis), we now express the same quantities additionally in the K-basis. These are not all independent:  some are related by purifications pairwise, namely $Q^{(4)}_1$ with $Q^{(4)}_2$, $Q^{(4)}_4$ with $Q^{(4)}_5$, and $Q^{(4)}_6$ with $Q^{(4)}_7$, as is most easily manifest from the K-basis representation, where the purification amounts to simply swapping 2 with 5, 125 with 512, and 1 with 5, respectively.  Note that, as advertised above, the number of terms in the I-basis changes only by 1 under the latter two purifications.

% ---------------------- (N=5)
\subsubsection*{$\N=5$:}

\begin{table}[htbp]
\begin{center}
\scriptsize
\begin{tabular}{| c | c |l  | }
\hline
Relation & \makecell{Basis}  & Primitive Information Quantity  \\ 
\hline
\hline
 SA$_{(1,1)}$
 &  S (3) &  $\S_{1}+\S_{2}-\S_{12}$ \\ 
 &  \shadeI{I (1)} & \shadeI{$\I_{12}$}  \\
 &  \shadeK{K (1)} & \shadeK{$\k_{12}^{(5)}$ } 
\\  \hline  
 MMI$_{(1,1,1)}$
 &  S  (7) & $-\S_{1}-\S_{2}-\S_{3}+\S_{12}+\S_{13}+\S_{23}-\S_{123}$   \\ 
&  \shadeI{I (1)} & \shadeI{$-\I_{123}$}  \\
 &  \shadeK{K (3)} & \shadeK{$\k_{1234}^{(5)}+\k_{1235}^{(5)}+\k_{1236}^{(5)}$ } 
\\  \hline 
 MMI$_{(1,2,2)}$
 &  S (7)  & $-\S_{1}-\S_{23}-\S_{45}+\S_{123}+\S_{145}+\S_{2345}-\S_{12345}$  \\   
&  \shadeI{I (9)} & \shadeI{$-\I_{124}-\I_{125}-\I_{134}-\I_{135}+ \I_{1234} + \I_{1235} + \I_{1245}+ \I_{1345}- \I_{12345}$}  \\
 &  \shadeK{K (5)} & \shadeK{$\k_{1246}^{(5)}+\k_{1256}^{(5)}+\k_{1346}^{(5)}+\k_{1356}^{(5)}+\k_{123456}^{(5)}
$ } 
\\ \hline 
$Q^{(5)}_1$
 &  S  (11) &
   $ - \S_{12} - \S_{23} - \S_{34}  - \S_{45} - \S_{15} + \S_{123} + \S_{234} + \S_{345} + \S_{145} + \S_{125}  - \S_{12345}$ \\ 
 &  \shadeI{I  (11)} & \shadeI{$-\I_{124} - \I_{134} - \I_{135} - \I_{235} - \I_{245} + \I_{1234} + \I_{1235} + \I_{1245} + \I_{1345} + \I_{2345} - \I_{12345}$}  \\
 &  \shadeK{K  (6)} & \shadeK{$\k_{1246}^{(5)} + \k_{1346}^{(5)} + \k_{1356}^{(5)} + \k_{2356}^{(5)} + \k_{2456}^{(5)} +2 \k_{123456}^{(5)}  $ } 
\\  \hline 
$Q^{(5)}_2$
 &  S  (16)  &  \makecell{$- \S_{12}- \S_{13} - \S_{14} - \S_{23} - \S_{25} - \S_{45}+2\S_{123} + \S_{124} + \S_{125} + \S_{134} + \S_{145} + \S_{235} + \S_{245}   $\\$- \S_{1234} - \S_{1235} - \S_{1245}$}  \\ 
  &  \shadeI{I  (7)} & \shadeI{$- \I_{124} - \I_{125} - \I_{135} - \I_{234} + \I_{1234} + \I_{1235} + \I_{1245} $}  \\
 &  \shadeK{K (7)} & \shadeK{$ \k^{\5}_{1246} + \k^{\5}_{1256} + \k^{\5}_{1345} + \k^{\5}_{1356} + \k^{\5}_{2345} + \k^{\5}_{2346} + 3 \k^{\5}_{123456} $ } 
\\  \hline 
$Q^{(5)}_3$
 &  S   (19) &   \makecell{$ - \S_{12} - \S_{13} - \S_{14} - \S_{25}- \S_{35} - \S_{45} + \S_{123}  + \S_{124} + \S_{125} + \S_{134} + \S_{135} + \S_{145} + \S_{235} $\\$  + \S_{245} + \S_{345}  - \S_{234} - \S_{1235} - \S_{1245} - \S_{1345}$}  \\ 
 &  \shadeI{I (7)} & \shadeI{$- \I_{125} - \I_{135} - \I_{145} - \I_{234} + \I_{1235} + \I_{1245} + \I_{1345} $}  \\
 &  \shadeK{K (7)} & \shadeK{$ \k^{\5}_{1234} + \k^{\5}_{1256} + \k^{\5}_{1356} + \k^{\5}_{1456} + \k^{\5}_{2345} + \k^{\5}_{2346} +3 \k^{\5}_{123456} $ } 
\\  \hline 
$Q^{(5)}_4 $
 &  S   (16)  &   \makecell{$- \S_2 - \S_3 - \S_4 - \S_5 - \S_{12} - \S_{13} + \S_{23} + \S_{45} + \S_{123} + \S_{124} + \S_{125} + \S_{134} + \S_{135} - \S_{145}$\\$ - \S_{1234} - \S_{1235}$}  \\ 
 &  \shadeI{I (6)} & \shadeI{$- \I_{123} - \I_{145} - \I_{234} - \I_{235} + \I_{1234} + \I_{1235} $}  \\
 &  \shadeK{K  (8)} & \shadeK{$ \k^{\5}_{1236} + \k^{\5}_{1245} + \k^{\5}_{1345} + \k^{\5}_{1456}  + 2 \k^{\5}_{2345}  + \k^{\5}_{2346} + \k^{\5}_{2356} + 2 \k^{\5}_{123456} $ } 
\\  \hline 
$Q^{(5)}_5 $
 &  S   (22) &  \makecell{$ - 2\S_{12} - 2\S_{13} - \S_{14} - \S_{15} - \S_{23}  - 2\S_{24} - 2\S_{35}  - \S_{45}+ 3\S_{123} + 3\S_{124} + \S_{125} + \S_{134} $\\$  + 3\S_{135}+ \S_{145} + \S_{234} + \S_{235} + \S_{245} + \S_{345} - 2\S_{1234} - 2\S_{1235} - \S_{1245} - \S_{1345}$}   \\ 
  &  \shadeI{I  (10)} & \shadeI{$- \I_{123} -2 \I_{125} - 2 \I_{134} - \I_{145} - \I_{234} - \I_{235} + 2\I_{1234} +2 \I_{1235} + \I_{1245} + \I_{1345} $}  \\
 &  \shadeK{K  (10)} & \shadeK{$\k^{\5}_{1236} + \k^{\5}_{1245} + 2 \k^{\5}_{1256} + \k^{\5}_{1345} + 2 \k^{\5}_{1346} + \k^{\5}_{1456}  + 2 \k^{\5}_{2345}    + \k^{\5}_{2346} + \k^{\5}_{2356} + 6 \k^{\5}_{123456}$ } 
\\  \hline 
\end{tabular}
\end{center}
\caption{Sign-definite (non-negative) primitive information quantities for $\N=5$ in the S, I, and K bases.  The number of terms for each quantity is indicated in parentheses in the second column  and is summarized in  Table~\ref{termnumberN5}.}
\label{tab:summaryN5}
\end{table}

Next we examine the $\N=5$ sign-definite quantities.
For a 5-party system, the entropy vector now has 31 components, which in the K-basis corresponds to $\binom{6}{2} = 15$ Bell pairs, $\binom{6}{4} = 15$ PT$_4$'s, and $\binom{6}{6}=1$ PT$_6$ involving all parties. Analogous to \eqref{entropy3} and \eqref{entropy4}, we can express each $\S_\pI$ in terms of $\k^{(5)}_\gI$ for various $\gI \subseteq \{1,\ldots,6\}$, where $\gI$ has an even number of elements. 

The set of all sign-definite primitive information quantities for $\N=5$  \cite{Bao:2015bfa,Cuenca:2019uzx} (up to permutations and purifications) are listed in Table \ref{tab:summaryN5}.
Here too some of the expressions are somewhat simpler in the I-basis for different instances under purifications, particularly 
 MMI$_{(1,2,2)}$ and $Q_1^{(5)}$, as pointed out in \sref{s:infoQ}.
 
%~~~~~~~~~~~~~~~~~~~~~~~~~~~~~~~~
\section{Method of Contraction}\label{a:contraction}

In this section, we will prove a weaker form of Theorem~\ref{positivethm}, which states that in a holographic entropy inequality expressed in the K-basis, all the coefficients $\l_\gI$ in \eqref{thmeqn} are positive. First, we will delineate the notations and definitions that will be employed in the proof. We will then state and prove the main result. In both subsections, we will follow closely both the notation as well as the proof strategy employed in proving Theorem 8 in  \cite{Bao:2015bfa}.

\subsection{Conventions and Notation}

Consider a holographic entropy inequality involving $\N$ parties of the form\footnote{\ 
In our conventions, the purifying party does not appear in any of the terms.}
\begin{align}\label{genineq}
	\sum_{l=1}^L \a_lS(\pI_l) \geq \sum_{r=1}^R \b_rS(\pJ_r)\ ,
\end{align}
where $\a_1,\ldots,\a_L,\b_1,\ldots,\b_R > 0$ are positive coefficients and $\pI_1,\ldots,\pI_L,\pJ_1,\ldots,\pJ_R \subseteq \{1,\ldots,\N\}$ are the corresponding subsystems. We want to rewrite this entropy inequality in the K-basis. Recalling there are $\binom{\N+1}{2s}$ possible \PTt{2s}'s, we can label the various PTs by using the collective index $\Gamma$, which consists of even-numbered polychromatic subsystems in $[\N+1]$, where $n_\gI = 2s$ for a PT$_{2s}$. Then we can express the entropy $S(\pI_l)$ and $S(\pJ_r)$ in terms of the PTs in the following manner:
\begin{align}
\begin{split}
	S(\pI_l) &=  \sum_{\Gamma} |C_\G(\pI_l\cap\Gamma)|\k^\n_{\Gamma} \\
	 S(\pJ_r) &= \sum_{\Gamma} |C_\G(\pJ_r \cap\Gamma)|\k^\n_{\Gamma}\ ,
\end{split}
\end{align}
where $|C_\G(\CX)| \equiv \min(n_\CX,n_{\G\setminus\CX})$.\footnote{\
This notation is chosen since we can view the \PTt{2s} involving $\G$ as a star graph, where all $2s$ parties in $\G$ are boundary vertices that are joined together by $2s$ edges to a common central vertex (see Fig.~\ref{f:PTfig}). $C_\G(\CX)$ then denotes the minimal cut of $\CX$, which bipartitions the vertex set of graph such that $\CX$ is only in one partition and $\bar \CX$ is in the other, and the number of edges crossing between the two partitions is minimized. In the case where $\CX = \pI_l \cap\G$ or $\pJ_r\cap \G$, the magnitude of the minimal cut agrees with the definition of $|C_\G(\CX)|$ given in the text.} Substituting these linear combinations into \eqref{genineq}, we get the entropy inequality in terms of the K-basis:
\begin{align}
 \sum_{\Gamma} \left[ \sum_{l=1}^L  |C_\G(\pI_l\cap\Gamma)|\a_l  - \sum_{r=1}^R  |C_\G(\pJ_r \cap\Gamma)|\b_r \right] \k^\n_{\Gamma} \geq 0\ .
\end{align}
For instance, if $\N=3$ and \eqref{genineq} is MMI, then this procedure yields precisely the K-basis expression of MMI$_{(1,1,1)}$ given in Table~\ref{tab:summaryN3}. It follows then that the positivity of $\lambda_\gI$ given in \eqref{thmeqn} of Theorem~\ref{positivethm} is equivalent to the condition that
\begin{align}\label{ineq2}
	\sum_{l=1}^L  |C_\G(\pI_l\cap\Gamma)|\a_l  - \sum_{r=1}^R  |C_\G(\pJ_r \cap\Gamma)|\b_r \geq 0
\end{align}
for any valid entropy inequality. We will now prove this statement using the method of contraction.

\subsection{Proof via Contraction}

While we will not be able to prove Theorem~\ref{positivethm} in its fully general form as stated using the method of contraction, we will be able to show that \eqref{ineq2} is true for \emph{any holographic entropy inequality with a contraction map}. This contraction map was introduced in \cite{Bao:2015bfa}, and it was subsequently confirmed that in fact all known holographic entropy inequalities have a contraction map. If it turns out that {\it all} holographic entropy inequalities possess a contraction map, then the theorem below can be viewed as an alternative proof of Theorem~\ref{positivethm}. 

Before we can describe the contraction map, however, we need to introduce the notion of occurrence vectors, which is defined as
\begin{align}
\begin{split}
	\ovx{i} &\equiv (i \in \pI_l)_{l=1}^L \in \{0,1\}^L \\
	\ovy{i} &\equiv (i \in \pJ_r)_{r=1}^R \in \{0,1\}^R\ ,
\end{split}
\end{align}
where $i=1,\ldots,\N+1$. Note that in our conventions the purifying party does not appear in any of the terms in the entropy inequality \eqref{genineq}, so both $\ovx{\N+1}$ and $\ovy{\N+1}$ are the zero vector. Lastly, we define the weighted Hamming norm $\| \vec v\,\|_\a \equiv \sum_{l=1}^L \a_l|v_l|$, where $v_l$ denotes the $l$-th component of a vector $\vec v$. Similarly, we define $\|\vec v\,\|_\b \equiv \sum_{r=1}^R \b_r|v_r|$. We can finally now state the theorem we want to prove.
\begin{theorem}
Consider the holographic entropy inequality
\begin{align}
	\sum_{l=1}^L \a_lS(\pI_l) \geq \sum_{r=1}^R \b_rS(\pJ_r)\ .
\end{align}
We may encode the subsystems $\pI_l$ on the left-hand-side using $\ovx{1},\ldots,\ovx{\N+1}$ defined above, and the subsystems $\pJ_r$ on the right-hand-side using $\ovy{1},\ldots,\ovy{\N+1}$ defined above. Let $f:\{0,1\}^L \to \{0,1\}^R$ be a contraction map with respect to the weighted Hamming norm, i.e.
\begin{align}
	\left\|f(\vec x) - f(\vec x\,')\right\|_\b \leq \left\|\vec x - \vec x\,'\right\|_\a \quad \forall\ \vec x,\vec x\,' \in\{0,1\}^L\ .
\end{align}
If $f\big(\ovx{i}\big) = \ovy{i}$ for all $i=1,\ldots,\N+1$, then
\begin{align}\label{goal}
	\sum_{l=1}^L  |C_\G(\pI_l\cap\Gamma)|\a_l  \geq \sum_{r=1}^R  |C_\G(\pJ_r \cap\Gamma)|\b_r\ ,
\end{align}
where $|C(\CX)| \equiv \min(n_\CX,n_{\G\setminus\CX})$, and $\Gamma$ is the collective polychromatic index denoting the PTs.
\end{theorem}
\begin{proof}
For any fixed $\pI_l$, note that given any $i \in\Gamma$, $x^i_l = 1$ is equivalent to $i \in \pI_l$. Using this fact, a moment's thought yields
\begin{align}
	|C(\pI_l\cap \Gamma)| = \sum_{i \in\Gamma} \left|x^\i_l - x^{\{\G\}}_l\right|\ ,
\end{align}
where we defined $\ovx{\G} \in \{0,1\}^L$ to have $0$ in the $l$-th component if $n_{\pI_l \cap \Gamma} \leq s$ and 1 otherwise.\footnote{\
As in the previous footnote, our notation is again inspired by the star graph representation of the PT involving $\G$. In this graph, the central vertex lies in the minimal cut of $\pI_l \cap \G$ if and only if $n_{\pI_l \cap\G} > s$. In this sense, $\ovx{\G}$ can be viewed as the occurrence vector for the central vertex, with $x^{\{\G\}}_l = 1$ if and only if the central vertex lies in the minimal cut of $\pI_l \cap \G$.} Similarly, we have
\begin{align}
	|C(\pJ_r\cap \Gamma)| = \sum_{i \in \Gamma} \left|y^\i_r - y^{\{\G\}}_r\right|\ ,
\end{align}
where we defined $\ovy{\G} \in \{0,1\}^R$ to have $0$ in the $r$-th component if $n_{\pJ_r \cap\G} \leq s$ and $1$ otherwise. Putting everything together, we get
\begin{align}\label{ineq}
\begin{split}
	\sum_{l=1}^L |C(\pI_l\cap \Gamma)|\a_l &= \sum_{l=1}^L\sum_{i\in\Gamma} \a_l \left|x^\i_l - x^{\{\G\}}_l\right| = \sum_{i\in\Gamma} \left\|\ovx{i}-\ovx{\G}\right\|_\a \\
	&\geq \sum_{i\in\Gamma} \left\|f\big(\ovx{i}\big) - f\big(\ovx{\G}\big)\right\|_\b = \sum_{r=1}^R \b_r\sum_{i\in\Gamma} \left|y^\i_r - f\big(\ovx{\G}\big)_r\right|\ .
\end{split}
\end{align}
If $f\big(\ovx{\G}\big)_r = 1$, then $\sum_{i\in\Gamma}\big|y^\i_r - f\big(\ovx{\G}\big)_r\big| = 2s- n_{\pJ_r\cap\G}$, whereas if $f\big(\ovx{\G}\big)_r = 0$, then $\sum_{i\in\Gamma}\big|y^\i_r - f\big(\ovx{\G}\big)_r\big| = n_{\pJ_r\cap\G}$. In either case, this is greater than or equal to $|C(\pJ_r\cap \G)|$, so
\begin{align}
	\sum_{i\in\Gamma}\left|y^\i_r - f\big(\ovx{\G}\big)_r\right| \geq |C(\pJ_r\cap\Gamma)| = \sum_{i\in\Gamma}\left|y^\i_r - y^{\{\G\}}_r\right|\ .
\end{align}
Substituting this into \eqref{ineq}, we obtain
\begin{align}
\begin{split}
	\sum_{l=1}^L |C(\pI_l\cap \Gamma)|\a_l &\geq \sum_{r=1}^R \b_r\sum_{i\in\Gamma}\left|y^\i_r - y^{\{\G\}}_r\right| = \sum_{r=1}^R |C(\pJ_r\cap \Gamma)|\b_r\ ,
\end{split}
\end{align}
completing the proof.

\end{proof}

%~~~~~~~~~~~~~~~~~~~~~~~~~~~~~~~~

%\bibliography{entropy-bib}
%\bibliographystyle{JHEP}
\providecommand{\href}[2]{#2}\begingroup\raggedright\endgroup

\end{document}